\def\year{2023}\relax

\documentclass[letterpaper]{article} % DO NOT CHANGE THIS
\usepackage{aaai22}  % DO NOT CHANGE THIS
\usepackage{times}  % DO NOT CHANGE THIS
\usepackage{helvet}  % DO NOT CHANGE THIS
\usepackage{courier}  % DO NOT CHANGE THIS
\usepackage[hyphens]{url}  % DO NOT CHANGE THIS
\usepackage{graphicx} % DO NOT CHANGE THIS
\usepackage{natbib}  % DO NOT CHANGE THIS AND DO NOT ADD ANY OPTIONS TO IT
\usepackage{caption} % DO NOT CHANGE THIS AND DO NOT ADD ANY OPTIONS TO IT
\DeclareCaptionStyle{ruled}{labelfont=normalfont,labelsep=colon,strut=off} % DO NOT CHANGE THIS
\usepackage{todonotes}
\usepackage{amsmath, amssymb, amsthm}
\usepackage{paralist}
\usepackage{xspace}
\usepackage{tikz}
\usetikzlibrary{calc,shapes,arrows}
\usepackage{thmtools, thm-restate}
\usepackage{stmaryrd}
\usepackage{algorithm}
\usepackage[noend]{algpseudocode}

\usepackage{caption}

\title{Monitoring Arithmetic Temporal Properties on Finite Traces}
\author{Paolo Felli,\textsuperscript{\rm 1}
Marco Montali,\textsuperscript{\rm 2} 
Fabio Patrizi,\textsuperscript{\rm 3}
Sarah Winkler\textsuperscript{\rm 2}%\thanks{This work is partially supported by the UNIBZ projects x.}
}

\affiliations{
\textsuperscript{\rm 1}Universit{\`{a}} di Bologna -- Bologna -- Italy, paolo.felli@unibo.it\\
\textsuperscript{\rm 2}Free University of Bozen-Bolzano -- Bolzano -- Italy, {\{montali,winkler\}@inf.unibz.it}\\
\textsuperscript{\rm 3}Sapienza Universit{\`{a}} di Roma  -- Rome -- Italy, patrizi@diag.uniroma1.it\\
}

\newcommand{\PS}{\text{\textsc{ps}}}
\newcommand{\PV}{\textsc{pv}}
\newcommand{\CS}{\textsc{cs}}
\newcommand{\CV}{\textsc{cv}}

\newcommand{\mc}[1]{\mathcal{#1}}
\renewcommand{\vec}[1]{\overline{#1}}

\newcommand{\ints}{\mathit{int}}
\newcommand{\rats}{\mathit{rat}}

\newcommand{\hist}{h}
\newcommand{\Cinit}{\varphi_{\mathit{init}}}
\newcommand{\combine}[2]{[{#1}\,{\circledast}\,{#2}]} % find nicer notation?
\newcommand{\CC}{\mc C}
\newcommand{\U}{\mathrel{\mathsf{U}}} % LTL operator
\newcommand{\G}{\mathsf{G}\xspace}
\newcommand{\F}{\mathsf{F}\xspace}
\newcommand{\X}{\mathsf{X}\xspace}
\newcommand{\wX}{\mathsf{X}_w\xspace}
\newcommand{\sX}{\mathsf{X}_s\xspace}
\newcommand{\LL}{\mathcal L}

\newcommand{\NFA}[1][\psi]{\mathcal N_{#1}}
\newcommand{\DFA}[1][\psi]{\mathcal D_{#1}}
\newcommand{\QQ}{\mathcal Q}
\newcommand{\CG}{\textup{CG}} % constraint graph
\newcommand{\update}{\mathit{update}} % update operation in CG construction
\newcommand{\back}{\mathit{back}} % subscript for translated \psi
\newcommand{\constr}{\mathit{constr}} % constraints in symbol
\newcommand{\goto}[1]{\mathrel{\raisebox{-2pt}{$\xrightarrow{#1}$}}}
\newcommand{\gotos}[1]{\to_{#1}^*}
\newcommand{\inn}{\,{\in}\,} % less spacious \in relation 
\newcommand{\cgnode}[2]{{#1}\nodepart{two}{#2}}
\newcommand{\inquotes}[1]{\ensuremath{\text{\textgravedbl}\!{#1}\!\text{\textacutedbl}}}
\newcommand{\prev}{{\mathit{pre}}}
\newcommand{\curr}{{\mathit{cur}}}
\newcommand{\cur}[1]{{#1}_\curr}
\newcommand{\pre}[1]{{#1}_\prev} % {\bar{#1}}

\newcommand{\m}[1]{\mathsf{#1}}
\newcommand{\sat}{\textsc{FSat}}
\newcommand{\unsat}{\textsc{FUns}}
\newcommand{\domino}[4]{\text{\tikz[baseline=-0.5ex]{\node[scale=.7, inner sep=0pt]{$%
\left\{\begin{array}{@{\,}l@{=}l@{\,}}{#1}&{#2}\\{#3}&{#4}\\\end{array}\right\}$}}}}

\tikzstyle{state}=[draw, circle, inner sep=1.5pt, line width=.7pt, scale=.6]
\tikzstyle{edge}=[draw, ->, line width=.5pt]
\tikzstyle{action}=[scale=.55]
\tikzstyle{caption}=[scale=.9]
\tikzstyle{node} = [draw,rectangle split, rectangle split parts=2,rectangle split horizontal, rectangle split draw splits=true, inner sep=3pt, scale=.65, rounded corners=2pt]
\tikzstyle{goto} = [->]
\tikzstyle{action}=[scale=.6, black]
\tikzstyle{final}=[double]
\tikzstyle{pscolor}=[fill=green!30]
\tikzstyle{cscolor}=[fill=cyan!30]
\tikzstyle{pvcolor}=[fill=red!30]
\tikzstyle{cvcolor}=[fill=orange!30]

\newtheorem{theorem}{Theorem}

\newtheorem{example}[theorem]{Example}
\theoremstyle{definition}
\newtheorem{definition}[theorem]{Definition}
\newtheorem{lemma}[theorem]{Lemma}
\newtheorem{corollary}[theorem]{Corollary}

\newcommand{\altlf}{A\ltlf}
\newcommand{\ltlf}{\textsc{LTL}$_f$\xspace}

\newcommand{\rvltl}{\textsc{RV-LTL}\xspace}

\newcommand{\NEW}[1]{{\color{blue}{#1}}}
\renewcommand{\NEW}[1]{#1}

\newcommand{\longversion}[2]{#1} % first argument for appendix version, second otherwise

\begin{document}
\maketitle

\begin{abstract}
We study monitoring of linear-time arithmetic properties against finite traces generated by an unknown dynamic system. The monitoring state %of the property 
is determined by considering at once the trace prefix seen so far, and %together with
all its possible finite-length, future continuations. This makes monitoring at least as hard as satisfiability and validity. Traces consist of finite sequences of assignments of a fixed set of variables to numerical values.
% 
%AVOIDING EMPHASIS ON NEW LOGIC
%Properties are specified in the new, general temporal logic \altlf, which subsumes several logics from the literature, combining \ltlf (LTL on finite traces) with linear arithmetic constraints that may carry \emph{lookahead}, i.e., variables may be compared over multiple instants of the trace. 
Properties are specified in a logic we call \altlf, %which subsumes several logics from the literature, 
combining \ltlf (LTL on finite traces) with linear arithmetic constraints that may carry \emph{lookahead}, i.e., variables may be compared over multiple instants of the trace. 
While the monitoring problem for this setting is undecidable in general, we show decidability for (a) properties without lookahead, and (b) properties with lookahead that satisfy the abstract, semantic condition of \emph{finite summary}, studied before in the context of model checking. We then single out concrete, practically relevant classes of constraints guaranteeing finite summary. %that guarantee finite summary, and are hence decidable.
Feasibility %of our approach 
is witnessed by a prototype implementation. 
\end{abstract}

\section{Introduction}

%In artificial intelligence, dynamic systems 
Dynamic systems in AI are often constituted by autonomous agents and heterogeneous components whose internal specification is either unknown %, not available, 
or not accessible, so that well-known techniques to ascertain their correctness at design-time, such as model checking and testing, are not applicable. % in this setting. 
This calls for approaches to check %the satisfaction of 
desired
properties at runtime, by \emph{monitoring} the executions of the black-box system under scrutiny. 
A widely known, solid approach to monitoring is that of \emph{runtime verification} (RV), where given a logical property, 
the monitor emits a provably correct verdict~\cite{LeuS09}. 
%S: avoid repetition with below
% As monitored executions consist of evolving, finite-length sequences of states, specification languages have a linear-time semantics. 
Formalisms based on LTL and its extensions have indeed been extensively applied when monitoring multi-agent systems \cite{DaTY18}, software components~\cite{LeuS09} and business processes~\cite{LMMR15,MaggiFDG14}. 

\citeauthor{LeuS09} (\citeyear{LeuS09}) highlight two essential semantic desiderata  %to be considered when defining 
of linear-time monitors. First, a monitored trace has a finite length, which calls for a finite-trace semantics. Second, a trace is the prefix of an unknown, full trace. Thus, the verdict of the monitor should be \emph{anticipatory}~\cite{BartocciFFR18}, i.e., depend not only on the data seen so far, but also on its (infinitely many) possible continuations.

In this spectrum, we build on %the solid body of work where monitoring formulae are expressed in 
the widely studied logic \ltlf, or LTL over finite traces~\cite{DeGV13}, where both the monitored prefixes and their infinitely many suffixes have a finite (yet unbounded) length. 
% \todo{ref to LDL$_f$?}
Specifically, we focus on the more sophisticated setting of \rvltl by \citeauthor{BaLS10} (\citeyear{BaLS10}), where the verdict of the monitor of a property $\psi$ at each point in time %maps the monitored formula to 
is one of four possible values: \emph{currently satisfied} ($\psi$ is satisfied now but may be violated in a future continuation), \emph{permanently satisfied} ($\psi$ is satisfied now and will necessarily stay so), and the two complementary values of \emph{permanent} and \emph{current violation}. 
Thus, %is immediately shows that 
\rvltl monitoring is at least as hard as satisfiability and validity. % of the adopted logical formalism. 
%Differently from the case of LTL on infinite traces, every \ltlf formula is monitorable, and its.
For propositional \ltlf, an \rvltl monitor can be constructed from the deterministic finite-state automaton (DFA) corresponding to the property, by labeling each DFA state by an \rvltl value, depending on whether it is final  and can reach final and non-final states \cite{MMWV11, DDMM22}. 
% Notably, this approach has been adopted not only to monitor black-box dynamic systems, but also by \citeauthor{DFIP20} (\citeyear{DFIP20}) to define non-Markovian rewards for MDPs.

Prior work on \ltlf monitoring like the above has often focused on propositional traces, where states are described by propositional interpretations. 
This coarse-grained representation creates a large abstraction gap w.r.t. dynamic systems, which often generate \emph{data-aware} traces where states contain richer objects from an infinite domain, e.g., strings or numbers. On the other hand, for richer logics partial methods have been devised,
but without guarantees that the monitoring task can be solved, e.g.~\cite{RegerCR15}.
Aiming to bridge this gap, we
\begin{inparaenum}[\it (i)]
\item introduce an extension of \ltlf, called \altlf, which combines linear-time operators with linear arithmetic constraints over data variables \emph{within and across states}, 
\item study how to construct automata-based \rvltl monitors for such formulae 
and 
\item identify decidable classes of properties, given that \altlf monitoring is in general undecidable.
%\todo{decidable vs solvable in entire intro}
\end{inparaenum}
\NEW{The next example illustrates the challenge that we aim to address.}

\newcommand{\res}{\mathit{r}}
\newcommand{\price}{\mathit{p}}
\newcommand{\bidder}{\mathit{b}}
\newcommand{\timer}{\mathit{t}}
\newcommand{\duration}{\mathit{T}}
\newcommand{\act}{\mathit{act}}
\begin{example}
\label{ex:intro}
\NEW{
Consider the monitoring of suspicious bidding patterns in an on-line auction.
The auction process can be depicted as a labeled transition system as follows:\\
\resizebox{\columnwidth}{!}
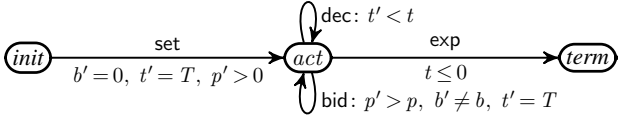
{
\begin{tikzpicture}[node distance=42mm,>=stealth']
\tikzstyle{action}=[scale=.55]
\tikzstyle{pstate}=[draw, rectangle, rounded corners, inner sep=3pt, line width=.7pt, scale=.6]
\node[pstate] (1)  {init};
\node[pstate, right of=1] (2) {$\act$};
\node[pstate, right of=2] (3) {term};
\draw[edge] (1) to node[above,action]
{$\m{set}$}
node[below,action]
{$\bidder'\,{=}\,0,\ \timer' \,{=}\,\duration,\ \price' \,{>}\,0$}(2);
\draw[edge] (2) to node[above,action] {$\m{exp}$} node[below,action] {$\timer \,{\leq}\,0$} (3);
\draw[edge, loop above] (2) to node[right,action, yshift=-2mm] {$\ \m{dec}\colon\timer' \,{<}\,\timer$} (2);
\draw[edge, loop below] (2) to node[right,action, yshift=2mm] {$\ \m{bid}\colon\price'\,{>}\,\price,\ \bidder'\,{\neq}\,\bidder,\ \timer' \,{=}\,\duration$} (2);
\end{tikzpicture}
}
The process maintains three variables: the current price $\price$ (i.e., the last bid); the id of the last bidder $\bidder$, and a timer $t$. In the edge labels, $'$ indicates the next value of the respective variable that is written by the transition. The transition $\mathsf{set}$ fixes $\price$ to some base price and initializes the timer to the duration $\duration$ of a round. Then,  either the timer is decreased by transition $\mathsf{dec}$, or a bidder increases the price, upon which the timer is reset to $\duration$.
Finally, $\mathsf{exp}$ terminates the auction when the timer expires.
(We assume that all variables not explicitly written in transition guards keep their value.)
With an additional variable $s$ for the state, the process can be encoded in ALTL$_f$ as follows:
\begin{equation}
\label{eq:transition:system}
(s\,{=}\,\mathit{init}) \wedge \G(\psi_{\m{set}} \vee \psi_{\m{dec}} \vee \psi_{\m{bid}} \vee \psi_{\m{exp}})
\tag{$\star$}
\end{equation}
where, e.g., $\psi_{\m{set}} = (s\,{=}\,\mathit{init} \wedge s'\,{=}\,\act \wedge \timer' \,{=}\,\duration\wedge\price' \,{>}\,0)$, and  other transitions are encoded similarly.}

\NEW{
During an auction, it is desirable to identify users that exhibit \emph{shilling behavior}, i.e., they drive up the auction price for the seller.
Given a user $u$, consider the following behavior patterns~\cite{XuC07}:
\begin{inparaenum}
\item[(OB)] \emph{Overbidding}: $u$ places an overbid (i.e.,increasing the price by at least 20\%) when the last bid happened a while ago. This can be expressed as $\mathit{OB}_u:= \F(\bidder'\,{=}\,u \wedge \timer\,{\leq}\,\epsilon \wedge \price'\,{\geq}\price\cdot 1.2)$, for some small $\epsilon$.
\item[(AU)] \emph{Aggressive underbidding}: \emph{underbidding} means that the price increase is less than 3\$, and \emph{aggressive} that the time gap between two bids is small (for simplicity, we assume here in the next state). This can be
expressed as $\mathit{AU}_u:= \G(\bidder\,{\neq}\,u \wedge \bidder\,{>}\,0 \to \price'\,{\geq}\,\price{+}3 \wedge\bidder\,{=}\,u)$.
% \item[(BS)] \emph{Bid sniping}: $u$ only bids after a previously fixed reserve price is reached, which can be
% expressed as $\mathit{BS}_u:= (\bidder\,{\neq}\,u)\U(\price\,{\geq}\,\res\wedge \bidder\,{=}\,u)$. 
\end{inparaenum}
% While (OB) and (AU) likely imply shilling, (BS) excludes it.
Both  (OB) and (AU) likely imply shilling.
The following is an example process execution (also called a \emph{trace}) with $\duration=2$.% and $\res=50$.
\\[1ex]
$\begin{array}{l|r@{\hspace{3mm}}r@{\hspace{3mm}}r@{\hspace{3mm}}r@{\hspace{3mm}}r@{\hspace{3mm}}r@{\hspace{3mm}}r@{\hspace{3mm}}r@{\hspace{3mm}}r@{\hspace{3mm}}r}
s &\mathit{init}&\act&\act&\act&\act & \act& \act& \act\\
\timer & 0   & 2   & 2   & 2   &  1  & 2   &  2  &  2  \\
\price & 0   & 10  & 30  & 32  & 32  & 36  & 40  & 50  \\
\bidder& 0   & 0   & 1   & 2   &  2  & 3   & 1   & 2   \\
\multicolumn{9}{c}{}\\[-1.3ex]
\mathit{OB}_2 & \CV & \CV & \CV & \CV & \CV & \CV & \CV & \PS \\
\mathit{AU}_2 & \CS & \CS & \CS & \CS & \CS & \CS & \PV & \PV \\
\mathit{shill_2} & \CS & \CS & \CS & \CS & \CS & \CS & \CV & \PS\\
% \mathit{BS}_4 & \CV & \CV & \CV & \CV & \CV & \CV & \CV & \CV &\PS \\
\end{array}$\\[1ex]
The lower part of this table shows monitoring results for the trace, namely 
the RV values obtained for observing
(OB) and (AU) for user 2, as well as for
$\mathit{shill_2} := \mathit{OB}_2 \vee \mathit{AU}_2$. % and (BS) for user 4, 
% always in conjunction with \eqref{eq:transition:system}.
The techniques presented in this paper show how to construct monitors for such properties fully automatically.}
\end{example}

Technically, we follow \citeauthor{DD07} (\citeyear{DD07})
and consider traces consisting of assignments of a fixed set of variables to
numeric values (which may encode complex data objects, e.g. strings).
% corresponding data objects (strings, reals, integers, etc.). 
A starting point is the approach in \cite{FMW22a}, where automata-based techniques are used to model check dynamic systems against an extension of \ltlf with linear arithmetic constraints, e.g., 
``the value of sensor $x$ differs from that of sensor $y$ by 5 units until it exceeds 25''.
However, that approach has two prohibitive technical limitations. 
The first one is operational: their automata are too weak for monitoring as they are faithful only for reachability of final, but not non-final states.
The second shortcoming is semantical: arithmetic conditions can only relate variables within the same state, so one cannot specify the (un)desired progression of assignments over instants in the trace. %, so that it is not possible to specify  
%This hampers the possibility of expressing properties referring to 
%the (un)desired progression of assignments within a trace, e.g. that ``two consecutive values carried by $x$ must always differ by at most 5 units''. 
For this reason \citeauthor{DD07} (\citeyear{DD07}) and, more recently, \citeauthor{GGG22} (\citeyear{GGG22}), allow properties to compare variables across states. 
Our logic \altlf follows this route, enabling variable \emph{lookahead}
to reference variable values in different instants within arbitrary linear arithmetic constraints, (cf. e.g. $p'\,{>}\,p$ in Ex.~\ref{ex:intro}). 
Our approach is thus substantially different from \cite{DD07}, which restricts to 
%they only deal with 
% \emph{monotonicity constraints} (MCs), that is,
variable-to-variable and variable-to-constant comparisons to guarantee decidability of satisfiability. % for their logic.
Instead, we allow arbitrary linear arithmetic constraints with lookahead, making satisfiability (and monitoring) of \altlf undecidable even for simple formulae.

Despite the rich research body on monitoring, 
there are few approaches that are anticipatory, feature full LTL$_f$, and incorporate arithmetic with lookahead~\cite{FalconeKRT21}.
The seminal LOLA approach~\cite{DAngeloSSRFSMM05} supports arithmetic with lookahead and finite traces, but is not anticipatory. 
Also
\citeauthor{FaymonvilleFSSS19} (\citeyear{FaymonvilleFSSS19}) and \cite{BasinKMZ15} have only bounded anticipation.
More related is
\textsf{MarQ} \cite{RegerCR15}, a strong, anticipatory RV tool for finite traces with relations and arithmetic, but its approach has no decidability guarantees.

We %address such operational 
overcome these limitations as follows: %presenting the following technical contributions.
\begin{inparaenum}[(1)]
\item we show that \altlf monitoring is decidable for linear arithmetic constraints without lookahead, and provide techniques to construct automata-based \rvltl monitors. % for such formulae.
\item we show, for formulae with lookahead, undecidability of \altlf satisfiability/monitoring.
\item To mitigate this, we study how the abstract property of \emph{finite summary}, introduced in \cite{FMW22a}, can be employed for \altlf formulae with lookahead that satisfy this property, proving  decidability of monitoring and providing a technique to construct automata-based \rvltl monitors. % for formulae with finite history, dealing with lookahead.
\item We use this general result to show decidability for different concrete classes of formulae enjoying finite summary, obtained by restricting either the language of constraints, or the way these  interact with each other via lookahead.
\item A prototype tool witnesses feasibility of our approach. 
\end{inparaenum}

\NEW{
The structure of the paper reflects Contributions (1)--(5) in this order, after
the introduction of \altlf and the monitoring problem 
and a section about the construction of automata for \altlf properties afterwards.}

\section{Preliminaries}

We consider the sorts $\ints$ and $\rats$ with domains $d(\ints)=\allowbreak\mathbb Z$ and $d(\rats) = \mathbb Q$.
For a set of variables $V$%that are of one of these sorts
, let $V_s\,{\subseteq}\,V$ be those in $V$ of sort $s$.
We define arithmetic constraints as:
\begin{definition}%[Constraints]
\label{def:constraint}
Given a set of variables $V$, expressions $e_s$ of sort $s$ and \emph{constraints} $c$ are defined as follows: \\
\begin{tabular}{r@{~}l@{\qquad}r@{~}l}
$e_s$ :=&\ $v_s\:\mid\:k_s\:\mid\:e_s + e_s\:\mid\:k_s \cdot e_s $ \\
$c$   :=&\ $e_s = e_s\:\mid\:e_s \neq e_s\:\mid\:e_s < e_s\:\mid\:e_s \leq e_s \mid {}$\\
        &\ $e_{\ints} \approx_n e_{\ints}\mid\:e_{\ints} \not\approx_n e_{\ints}$ 
\end{tabular}\\
where $k_s\in d(s)$ is a constant, $v_s \in V_s$ a variable of appropriate sort, and $\approx_n$ denotes 
congruence modulo $n\in \mathbb N$.
\end{definition}
\noindent
The set of all constraints over $V$ is denoted by $\CC(V)$.
E.g., $x\,{<}\,y\,{+}\,z$, and $x\,{-}\,y\,{\neq}\,2 $ are constraints
independent of the sort of $x, y, z$,
but $u \equiv_3 v + 1$ requires that $u$ and $v$ have sort $\ints$.
We also consider boolean formulas with constraints as atoms; these are in the realm of SMT with linear arithmetic, which is decidable and admits \emph{quantifier elimination}~\cite{Presburger29}:
if $\varphi$ is a formula with free variables $X \cup \{y\}$  
and constraint atoms,
there is some $\varphi'$ with free
variables $X$ that is logically equivalent to $\exists y. \varphi$, 
i.e., $\varphi'\,{\equiv}\,\exists y. \varphi$. 

From now on, let $V$ be a fixed, finite, non-empty set of variables
called \emph{state variables}.
An \emph{assignment} $\alpha$ maps every $v\inn V$ to a value $\alpha(v)$ in the domain of its sort. 
\NEW{
The set of \emph{variables with lookahead} ${\mathcal V}$ consists of all $v^{\prime\cdots\prime}$ such that $v\in V$ is decorated with $i$ prime symbols $'$, for $i\geq 0$, and we say that $v$ \emph{has lookahead $i$}.
Intuitively, $v^{\prime\cdots\prime}$ with $i$ primes refers to the value of $v$ looking $i$ instants into the future.
For instance, $v'$ refers to the value of $v$ in the next step, while $v''$ is the value of $v$ two steps ahead.
}
% so $v^0$ is the current value of $v$. We also write $v'$ for $v^1$.

\begin{definition}
Let $\LL$ be the set of properties $\psi$ defined by the following grammar,
where $c \in \CC(\mathcal V)$ is a constraint:
\[\psi :: = 
c \mid 
\psi {\wedge} \psi \mid \neg \psi \mid
\sX \psi \mid \psi \U \psi
\]
\end{definition}
When needed, we write $\LL(V)$ instead of $\LL$ to make explicit the set of variables in constraints.
The usual transformations apply, namely $\psi_1 {\vee} \psi_2 \equiv \neg (\psi_1 {\wedge} \psi_2)$, $\wX \psi \equiv \neg (\sX \neg \psi)$,
$\F \psi \equiv \top \U \psi$, $\G \psi \equiv \neg \F \neg \psi$;
moreover $\top \equiv (v = v)$ for some $v\,{\in}\,V$, and $\bot \equiv \neg \top$.
A property $\psi$ \emph{has lookahead $m$} if 
\NEW{the maximal lookahead of a variable in $\psi$ is $m$.
For instance, all properties in Ex.~\ref{ex:intro} are in this language; the lookahead of $\mathit{AU}_2$ is 1 but of $\timer \,{\leq}\,0$ it is 0.}
% for all variables $v^{\prime\cdots\prime}$ in $\psi$ with lookahead $i$ it holds that $i\leq m$.

Properties are evaluated over \emph{traces}: a trace $\tau$ of length $n\,{\geq}\,1$ is a finite sequence
$\alpha_0, \alpha_1, \dots, \alpha_{n-1}$ of assignments with domain $V$. 
We write $\tau(i)$ to denote $\alpha_i$, for $0\,{\leq}\,i\,{<}\,n$.
A constraint $c$ is \emph{well-defined} at instant $i$ of $\tau$ if $0 \leq i < n$ and all \NEW{variables with lookahead $j$ in $c$ satisfy $i\,{+}\,j < n$.
The upper part of the table in Ex.~\ref{ex:intro} is an example of a trace,
and e.g., $\price' \geq \price$ is well-defined at all but the last instant.}

\begin{definition}
\label{def:witness}
%Let $\tau$ be a trace of length $n$.
If an expression $e$ is well-defined at instant $i$ of a trace $\tau$ of length $n$, %$0 \leq i < n$,
its \emph{evaluation} $[\tau,i](e)$ at instant $i$ is:\\[1ex] % given by: 
\begin{tabular}{@{}r@{\,}l@{\:\:}r@{\,}l@{}}
$[\tau,i](k)$ &$= k$  &
$[\tau,i](e{+}e_2)$ &$=  [\tau,i](e_1){+}[\tau,i](e_2)$\\
$[\tau,i](v^{\prime\cdots\prime})$ &$= \tau(i{+}j)(v)$ &
$[\tau,i](k\cdot e)$&$=k\cdot[\tau,i](e)$
\end{tabular}\\[1ex]
where $v^{\prime\cdots\prime}$ has lookahead $j$.
Then, $\tau$ \emph{satisfies} $\psi \in \LL$, denoted 
$\tau \models \psi$, iff $\tau,0 \models \psi$ holds, 
where, for $0\leq i$:

\noindent
\begin{tabular}{@{}l@{\:}l@{}}
$\tau,i \models$&$ e_1 \odot e_2$ iff $0 \leq i < n$ and either\\
& -- $e_1 \odot e_2$ is not well-defined for $\tau$ and $i$, or \\
& -- $[\tau,i](e_1) \odot [\tau,i](e_2)$ holds
\\
$\tau,i \models$&$ \psi_1 \wedge \psi_2$ iff 
 $\tau,i \models \psi_1$ and $\tau,i \models \psi_2$\\
$\tau,i \models$&$ \neg \psi$ iff 
 $\tau,i \not\models \psi$\\
% $\tau,i \models$&$ \psi_1 \vee \psi_2$ iff 
%  $\tau,i \models \psi_1$ or $\tau,i \models \psi_2$\\
$\tau,i \models$&$ \sX\psi$ iff $i<n-1$ and $\tau,i{+}1 \models \psi$\\
% $\rho,i \models$&$\G\psi$ iff 
% $\rho,i \models \psi$ and ($i=n-1$ or
% $\rho,i{+}1\models \G\psi$)\\
$\tau,i \models$&$ \psi_1 \U \psi_2$ iff $i < n$ and either
 $\tau,i \models \psi_2$, or \\
 & $i\,{<}\,n{-}1$,
 $\tau,i \models \psi_1$ and
 $\tau,i{+}1\models \psi_1 \U \psi_2$
\end{tabular}
\end{definition}

\noindent
\NEW{The trace in Ex.~\ref{ex:intro} e.g. satisfies $\mathit{OB}_2$, its proper prefixes do not.}
Note that lookahead variables have a \emph{weak} semantics, in that a constraint $c$ holds if it contains a variable $v^{\prime\cdots\prime}$ that refers to an instant beyond the trace. 
%If one instead wants to assume a strict interpretation, this can be done by replacing
If instead  a strict semantics is desired, one can replace
% A strict requirement can be expressed by replacing 
$c$ by $c\wedge \X_s\dots\X_s\top$, where the number of $\X_s$ operators equals 
the lookahead of $c$.

We now define our main task, that is to monitor how the satisfaction of a given property changes along a trace, that is, by considering the trace (fragment) $\tau$ seen \emph{so far}.  
As customary~\cite{BaLS10}, we consider the set $RV = \{\PS, \CS, \CV, \PV\}$ of four distinct \emph{monitoring states}:
current satisfaction ($\CS$), permanent satisfaction ($\PS$), current violation ($\CV$) and permanent violation ($\PV$).

\begin{definition}%[Monitoring state]
\label{def:monitoring}
A property $\psi\in \LL$ is in
monitoring state $s \in RV$ after a trace $\tau$, written $\tau \models \llbracket \psi = s\rrbracket$, if
\begin{compactitem}
\item
$s = \CS$, $\tau \models \psi$, and $\tau\tau' \not\models \psi$ for some trace $\tau'$;
\item
$s = \PS$, $\tau \models \psi$, and $\tau\tau' \models \psi$ for every trace $\tau'$;
\item
$s = \CV$, $\tau \not\models \psi$, and $\tau\tau' \models \psi$ for some trace $\tau'$; %and
\item
$s = \PV$, $\tau \not\models \psi$, and $\tau\tau' \not\models \psi$ for every trace $\tau'$.
\end{compactitem}
\end{definition}

E.g., $\tau \models \llbracket \psi = \CS\rrbracket$ means that $\psi$ is currently true after $\tau$ but there exists a possible continuation of $\tau$ (i.e., a trace $\tau\tau'$) after which $\psi$ is false. 
After a trace, a property $\psi$ is in exactly one possible monitoring state.
\NEW{For instance, the given monitoring states in Ex.~\ref{ex:intro} are as defined in Def.~\ref{def:monitoring}.}
Given input $\tau$ and $\psi$,
the \emph{monitoring problem} %for traces with arithmetic 
is to compute the state $s\in RV$ s.t.  
$\tau \models \llbracket \psi = s\rrbracket$. It is \emph{solvable} if
one can construct a procedure for $\psi$ that computes the monitoring
state for any given trace. Unfortunately, we have that in general:
%

% %%
% Unfortunately, determining whether a trace is in a monitoring 
% state $s$ is undecidable in general, implying that the monitoring state
% cannot be determined.
% We call the monitoring problem for an $\LL$ property \emph{solvable} if
% one can construct a monitor for $\psi$ which can compute the monitoring
% state for any given a trace.
%%
\begin{theorem}
The monitoring problem is not solvable.
\end{theorem}

To see this, call a property $\psi \in \LL$ \emph{satisfiable} if 
$\tau \models \psi$ for some (non-empty) trace $\tau$.
%then we have the following negative result:
It can be shown 
\longversion{(cf.~Lem. \ref{lem:undecidable} in the appendix)}{}
that 
satisfiability of $\LL$ properties with lookahead 1 is undecidable, by a reduction from reachability in 2-counter machines.
Now, for a formula of the form $\psi=\sX\psi'$ and a trace $\tau_0$ of length 1,
$\psi'$ is satisfiable iff 
$\tau_0\models\llbracket\psi=\CV\rrbracket$, 
so satisfiability of $\psi'$ reduces to monitoring  $\tau_0$ against $\psi$.

\smallskip
In this paper we study when monitoring is solvable. 
% For simplicity, and without loss of generality (see below), we restrict to properties with lookahead no greater than 1.  Hence we use $v$ for $v^0$ and $v'$ for $v^1$ to simplify the notation.
First, we show that without loss of generality one can restrict to properties with lookahead 1: Indeed, a 
property with lookahead $m{>}1$ can be transformed into one with $m{=}1$ by adding fresh variables and extending to these the assignments in the trace %(in a specific way). 
in an appropriate way.
In general, $m{-}1$ fresh book-keeping variables are needed for each variable with lookahead $m{>}1$.
Intuitively, 
%up to $m-1$ additional book-keeping variables are used to backshift in the trace each variable with lookahead greater than 1, so that constraints can be rewritten 
constraints are then rewritten to mention only consecutive instants, 
%within the scope of 
using chains of next operators.
E.g., 
$\psi{=} \G (x''\,{>}\,x)$ is equivalent to
$\widehat\psi{=} \G (u'\,{>}\,x \wedge \wX(u\,{=}\,x'))$.
The trace 
$\tau\colon 
\{x\,{=}\,2\},\allowbreak 
\{x\,{=}\,0\},\allowbreak 
\{x\,{=}\,3\}$
satisfies $\psi$ and  
$\widehat\tau\colon 
\{x\,{=}\,2,\allowbreak u\,{=}\,0\},\allowbreak  
\{x\,{=}\,0,\allowbreak u\,{=}\,3\},\allowbreak 
\{x\,{=}\,3,\allowbreak u\,{=}\,0\}$ 
satisfies $\widehat\psi$, with $\widehat\tau$ obtained from $\tau$ by assigning to $u$ the subsequent value of $x$. 
%
%
% The technical details are omitted, but a
A detailed proof of the next result can be found in the appendix.

\begin{restatable}{lemma}{lemmamorelookahead}
\label{rem:look:ahead:only:1}
Let $\psi\in \LL(V)$ have lookahead $m$.
There is some $\psi_1\in \LL(V \cup X)$ with lookahead $1$ for a set of fresh variables $X$ such that for every trace $\tau$ over $V$ there is a trace $\widehat\tau$ over $V \cup X$ satisfying $\tau \models \psi$ iff $\widehat\tau \models \psi_1$, where $|\widehat\tau|=|\tau|$ and $\widehat\tau(i)$ agrees with $\tau(i)$ on $V$, for all $i$.
\end{restatable}
%%
%%For this reason, in the remainder of the paper we will wlog. restrict to properties with lookahead at most 1. 
%i.e., it is a property over variables $V \cup V^1$, for $V^1 = \{ v^1 \mid v\in V\}$.
%

\section{Automata for properties}
\label{sec:automata}

In the sequel, we will develop first monitoring techniques for $\LL$ properties
without lookahead, and then with lookahead 1.
% We develop our technique for monitoring $\LL$ properties first without lookahead and then with lookahead 1. 
Although the former case is simpler and could be addressed by an ad-hoc approach resembling techniques for propositional LTL$_f$~\cite{MMWV11}, we adopt a uniform approach for both cases.
To this end, we first present the %basic 
construction of an automaton that represents a given property $\psi$; this will be crucial in the later sections.
As a preprocessing step, two transformations are applied to the property:

\noindent
(1) We assume $\psi$ in negation normal form; so it may also contain $\vee$, $\G$, and $\wX$,  so far considered syntactic sugar. 
Note that for constraints with lookahead, negation cannot be pushed to the constraint level, e.g., $\neg (x'\,{=}\,x)$ is not equivalent to $x'{\neq}x$ as the latter is satisfied by any trace of length 1 but the former need not be.
In the construction, we thus treat constraints with and without negation separately. 
For clarity, we write $neg(c)$ to flip the comparison operator in $c$: 
$neg (t_1\,{=}\,t_2) := t_1\,{\neq}\,t_2$,
$neg (t_1\,{<}\,t_2) := t_2\,{\leq}\,t_1$,
$neg (t_1\,{\approx_n}\,t_2) := t_2\,{\not\approx_n}\,t_1$,
and vice versa, for all $t_1,t_2$.

\noindent
(2) %When monitoring $\psi$, 
We want a monitor to return the monitoring state for the trace (prefix) $\tau$ seen so far, without depending on future variable assignments.
However, the evaluation of a constraint with lookahead depends on the next, not yet seen, assignment (see Def.~\ref{def:witness}). 
To avoid this counterintuitive peculiarity and make the monitoring task clearer, we transform constraints with look\emph{ahead} 1 into constraints with look\emph{back} 1.\\
%\todo{some restructuring} 
To this end, we consider variables $\pre{v}$ and $\cur{v}$, for all $v\inn V$.
% This is done merely for clarity, as it simplifies the presentation and the monitoring task, as illustrated in the example below.
% Intuitively, fixed a property $\psi$, we wish the monitoring task to return the monitoring state for the given trace (fragment) $\tau$ seen so far, without depending on the assignment for the \emph{next} instant (which would be needed to evaluate constraints -- see Def.~\ref{def:witness}). 
%
% and perform two replacements in $\psi$:
For a constraint $c$, let $back(c)$ be defined as
(a) if $c$ has no lookahead, $back(c)$ is obtained from $c$ by replacing $v$ with $\cur{v}$, and
(b)  if $c$ has lookahead, $back(c) :=\wX \bar{c}$, with $\bar{c}$ obtained from $c$ by replacing $v$ with
$\pre{v}$ and $v'$ with $\cur{v}$; for all $v\inn V$.
% (a) all constraints $c$ without lookahead are replaced by the constraint obtained from $c$
% by substituting $v$ with $\cur{v}$; and
% (b) all constraints $c$ with with lookahead 1 are replaced
% by $\wX \bar{c}$, with $\bar{c}$ obtained from $c$ by substituting all $v$ with
% $\pre{v}$ and $v'$ with $\cur{v}$, for all $v\inn V$.
%
We denote the property where each constraint $c$ is replaced by $back(c)$ by $\psi_\back$.
Then $\psi$ and $\psi_\back$ are equivalent 
\longversion{(cf.~Lem. \ref{lem:lookback:lookahead} in the appendix)}{(see long version)}.
We write $V_{\prev} {=} \{\pre{v} \mid v\in V\}$ and the same for $V_{\curr}$. 
%$V_{\curr} {=} \{v_{\curr} \mid v\in V\}$.

\begin{example}
Property $\psi = \G(x'\,{>}\,x) \wedge \F (x\,{=}\,2)$ is equivalent to $\psi_\back = \G(\wX (\cur{x}\,{>}\,\pre{x})) \wedge \F (\cur{x}\,{=}\,2)$.
\end{example}

Given $\psi_\back$, we build an NFA $\NFA[\psi_\back]$ using an auxiliary function $\delta$, as in~\cite{DeGV13}.
Let $C$ be the set of constraints in $\psi_\back$, and $C^\pm = C \cup\{ neg(c) \mid c\inn C\}$;
and $\lambda$ be an auxiliary proposition used to mark the last element of a trace.
The alphabet of the NFA will be $\smash{\Sigma=2^{C^\pm}}$, i.e., a symbol is a set of constraints. For the construction
we also use its extension $\smash{\Sigma_\lambda=2^{C^\pm \cup \{\lambda,\neg\lambda\}}}$.
Let $\varsigma\inn \Sigma_\lambda$ be \emph{satisfiable} if $\{\lambda, \neg \lambda \}\not \subseteq \varsigma$, and the conjunction of constraints in $\varsigma$ is satisfiable.
The input of $\delta$ is a (quoted) property $\phi \in \LL\cup \{\top,\bot\}$ over $V_\prev \cup V_\curr$.
The output of $\delta$ is a set of pairs
$(\inquotes{\phi'},\varsigma)$ where $\phi' \in \LL\cup \{\top,\bot\}$ is again
a (quoted) property over $V_\prev \cup V_\curr$
%has the same type as $\phi$
% \todo[inline]{FP: what do you mean by ``type'' of $\delta$ here?
% (Perhaps I missed the def of type, but better to briefly recall. P: yes, I don't remember either. Not changing this just to be safe.\\
% S: true, was an imprecise formulation introduced to save space. expanded.}
and $\varsigma \in \Sigma_\lambda$. 
For two sets of such pairs $R_1$, $R_2$, let
$R_1 \owedge R_2 = \{ (\inquotes{\psi_1 \wedge \psi_2}, \varsigma_1 \cup \varsigma_2) \mid (\inquotes{\psi_1}, \varsigma_1) \inn R_1, (\inquotes{\psi_2}, \varsigma_2) \inn R_2\text{ and }\varsigma_1\cup \varsigma_2\text{ is satisfiable} \}$, where we simplify $\psi_1 \wedge \psi_2$ if possible; and let $R_1 \ovee R_2$ be defined in the same way, replacing con- with disjunction.

Intuitively, $(\inquotes{\phi'},\varsigma)\in\delta(\phi)$ expresses that when $\phi$ is the property to evaluate, if $\varsigma$ is the current symbol of the word $w\inn \Sigma^+$ that is read, then $\phi'$ is the (sub)property yet to satisfy to determine whether $w$ satisfies $\phi$. 

\begin{definition}
For $\psi \in \LL \cup \{\top, \bot\}$, $\delta$ is as follows:

\begin{tabular}{@{~}r@{~}l}
$\delta(\inquotes{\top})$ =& $\{(\inquotes{\top},\emptyset)\} \text{ and }\delta(\inquotes{\bot}) = \{(\inquotes{\bot},\emptyset)\}$\\
$\delta(\inquotes{c})$ =& $\{(\inquotes{\top},\{c\}),(\inquotes{\bot},\{neg(c)\})\}$  \\
$\delta(\inquotes{\neg c})$ =& $\{(\inquotes{\bot},\{c\}),(\inquotes{\top},\{neg(c)\})\}$\\
$\delta(\inquotes{\psi_1 \wedge \psi_2})$ =&  $\delta(\inquotes{\psi_1}) \owedge \delta(\inquotes{\psi_2})$\\
$\delta(\inquotes{\psi_1 \vee \psi_2})$ =& 
  $\delta(\inquotes{\psi_1}) \ovee \delta(\inquotes{\psi_2})$ \\
  $\delta(\inquotes{\sX \psi})$ =& 
  $\{(\inquotes{\psi},\{\neg \lambda\}), (\inquotes{\bot}, \{\lambda\})\}$ \\
  $\delta(\inquotes{\wX \psi})$ =& 
  $\{(\inquotes{\psi},\{\neg \lambda\}), (\inquotes{\top}, \{\lambda\})\}$\\
$\delta(\inquotes{\G \psi})$ =&
  $\delta(\inquotes{\psi}) \owedge \delta(\inquotes{\wX\G \psi})$ \\
$\delta(\inquotes{\psi_1 \U \psi_2})$ =&
  $\delta(\inquotes{\psi_2}) \ovee (\delta(\inquotes{\psi_1})
  \owedge \delta(\inquotes{\sX(\psi_1 \U \psi_2)}))$.
\end{tabular}
\end{definition}

While the symbol $\lambda$ is needed for defining $\delta$, we can omit it
from the NFA, and define $\NFA$ as follows:
\begin{definition}
\label{def:NFA}
Given a property $\psi \in \LL$, we define the NFA as
$\NFA\,{=}\,(Q, \Sigma, \varrho, q_0, Q_F)$ where $q_0\,{=}\,\inquotes{\psi_\back}$ is the initial state,
$Q_F = \{\inquotes{\top}, q_{+}\}$ is the subset of final states, %where $q_{+}$ is an additional final state,
 $q_{-}$ is an additional non-final state,
and
$Q$, $\varrho$ are the smallest sets such that $q_0, q_F, q_{+}, q_{-} \in Q$ and whenever 
$q\in Q\setminus\{q_{+}, q_{-}\}$ and $(q', \varsigma)\in \delta(q)$
then $q'\in Q$ and 
\begin{compactenum}[(i)]
\item if $\lambda \not\in \varsigma$ then
$(q, \varsigma \setminus\{\neg \lambda \}, q') \in \varrho$, and
\item
whenever $\lambda \in \varsigma$, if $q' = \inquotes{\top}$ then
$(q, \varsigma \setminus\{\lambda \}, q_{+}) \in \varrho$, and if
$q' = \inquotes{\bot}$ then
$(q, \varsigma \setminus\{\lambda\}, q_{-}) \in \varrho$.
\end{compactenum}
\end{definition}

A word $w = \varsigma_0, \varsigma_1, \dots, \varsigma_{n-1}\in \Sigma^+$ is \emph{well-formed} if there are no (negated) constraints in $\varsigma_0$ that mention $V_{\prev}$. 
We next define the notion of consistency to relate words $w\inn\Sigma^+$ as above (where each $\varsigma_i$ is a set of constraints) with those traces $\tau = \alpha_0, \alpha_1, \dots, \alpha_{n-1}$ s.t., at each step, the assignment $\alpha_i$ satisfies $\varsigma_{i}$, i.e., the traces fitting $w$. 
To that end, for assignments $\alpha$ and $\alpha'$ with domain $V$, define for all $v\in V$ the assignment $\combine{\alpha}{\alpha'}$ with domain $V_{\curr} \cup V_{\prev}$ as 
$\combine{\alpha}{\alpha'}(\pre{v}) = \alpha(v)$ and 
$\combine{\alpha}{\alpha'}(\cur{v}) = \alpha'(v)$.

\begin{definition}
A well-formed word $\varsigma_0, \varsigma_1,\dots, \varsigma_{n-1} \in \Sigma^+$ is \emph{consistent} with a trace $\alpha_0, \alpha_1, \dots, \alpha_{n-1}$ if $\alpha_0 \models \varsigma_0$ and
$\combine{\alpha_{i-1}}{\alpha_{i}} \models \varsigma_i$
for all $0 < i < n$.
\end{definition}

The key property of $\NFA$ is stated in the following theorem: a word $w$ is accepted iff all the traces captured by $w$ satisfy $\psi$.
The proof can be found in the appendix.

\begin{restatable}{theorem}{lemmaNFA}
\label{lem:automaton:acceptance}
Given $\psi\inn\LL$, a trace $\tau$ and a well-formed word $w\inn \Sigma^+$ consistent with $\tau$, $\mathcal{N}_{\psi_\back}$ accepts $w$ iff $\tau \models \psi$. 
\end{restatable}

However, $\NFA[\psi_\back]$ is not deterministic, which is inconvenient for monitoring.
Thus, we next build an equivalent DFA $\DFA[\psi_\back]$ (exponentially larger in the worst case), using a subset construction,
and a restricted alphabet:
Let $C_\curr \subseteq C$ be the subset of constraints in $\psi_\back$ that mention only $V_\curr$ (but not $V_\prev$).
Let the alphabet $\Theta \subseteq \Sigma$ consist of all
maximal satisfiable subsets of $C^\pm$, and 
$\Theta_\curr \subseteq \Sigma$ consist of all maximal satisfiable subsets of $C_\curr^\pm$. 

\begin{definition}
For  $\NFA[\psi_\back]\,{=}\,(Q, \Sigma, \varrho, q_0, Q_F)$,
let $\DFA[\psi_\back] = (\QQ, \Theta \cup \Theta_\curr, \Delta, \{q_0\}, \QQ_F)$ where
$\QQ = 2^Q$,
$\QQ_F$ consists of all $P\,{\subseteq}\,Q$ such that $P\,{\cap}\,Q_F\,{\neq}\,\emptyset$, and the transition function is given by
$\Delta(\{q_0\}, \varsigma_0) = \{q' \mid (q_0,\varsigma', q') \inn \varrho\text{ and }\varsigma'\,{\subseteq}\,\varsigma_0\}$ for all $\varsigma_0 \in \Theta_\curr$, and
$\Delta(P, \varsigma) = \{q' \mid q\in P\text{, }(q,\varsigma', q') \in \varrho\text{ and }\varsigma' \subseteq \varsigma\}$ for all $\varsigma \in \Theta$ and $P \neq \{q_0\}$.
\end{definition}

Note that $\Delta$ is defined either $(i)$ for the initial state $\{q_0\}$ and $\varsigma \in \Theta_\curr$, i.e., a constraint set not mentioning $V_\prev$, or $(ii)$ for 
a non-initial state and  $\varsigma \in\Theta$.
This distinction ensures that $\DFA[\psi_\back]$ accepts only well-formed words, and
$\DFA$ is \emph{deterministic} in that
for every $w \in \Theta_\curr \Theta^*$ there is a unique state
$P$ of $\DFA[\psi_\back]$ such that $\{q_0\} \gotos{w} P$.
%From now on we will nevertheless call $\DFA$ a DFA due to the following property:
% \noindent
In the appendix, we show that $\DFA$ is equivalent to $\NFA$ in the following sense:

\begin{restatable}{lemma}{lemmaDFANFA}
\label{lem:DFA:NFA}
Given a trace $\tau$,
$\DFA$ accepts a word consistent with $\tau$ iff $\NFA$ accepts a word consistent with $\tau$.
% note that these words are over different signatures
\end{restatable}

\begin{example}
\label{exa:lookahead}
The property $\psi = \G(x' \geq x) \wedge \F (x = 2)$ is
equivalent to $\psi_\back = \G\wX (\cur{x}\,{\geq}\,\pre{x}) \wedge \F (\cur{x}\,{=}\,2)$; it demands that $x$ is increasing, and at some point has value 2. 
We get the following NFA $\NFA[\psi_\back]$: 

%\begin{center}
\noindent
\resizebox{\columnwidth}{!}
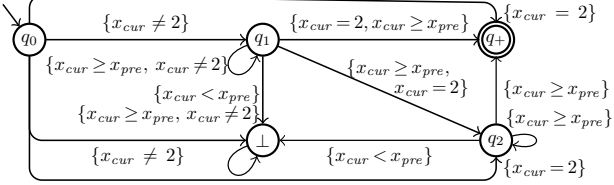
{
\begin{tikzpicture}[node distance=44mm]
 \node[state, minimum width=6mm] (0) {$q_0$};
 \node[state, minimum width=6mm, right of =0] (1) {$q_1$};
 \node[state, right of=1, minimum width=6mm, final] (top) {$q_{+}$};
 \node[state, below of=1, yshift=25mm, minimum width=6mm] (bot) {$\bot$};
 \node[state, right of=bot, minimum width=6mm] (2) {$q_2$};
\draw[edge] ($(0) + (-.3,.4)$) -- (0);
 \draw[edge] (0) -- node[action, above] {$\{\cur{x} \neq 2\}$} (1);
 \draw[edge, rounded corners] (0) -- ($(0)+(0,.5)$) --  ($(top) + (0,.42)$) -- node[action, right] {$\{ \cur{x}\,=\,2\}$} (top);
 \draw[edge, rounded corners] (0) |-  ($(2) + (0,-.45)$) -- node[action, right] {$\{ \cur{x}\,{=}\,2\}$} (2);
 \draw[edge, rounded corners] (0) |- node[action, below, near end] {$\{ \cur{x}\,\neq\,2\}$} (bot);
 \draw[edge] (1) -- node[action, right, near start, anchor=west, yshift=-2mm] {\begin{tabular}{r}$\{\cur{x}\,{\geq}\,\pre{x},\quad$\\$\cur{x}\,{=}\,2\}$\end{tabular}} (2);
 \draw[edge] (1) to node[action, above] {$\{\cur{x}\,{=}\,2,\cur{x}\,{\geq}\,\pre{x}\}$} (top);
 \draw[edge] (1) -- node[action, left, pos=.72, xshift=3mm] {\begin{tabular}{r}$\{\cur{x}\,{<}\,\pre{x}\}$\\$\{\cur{x}\,{\geq}\,\pre{x},\,\cur{x}\,{\neq}\,2\}$\end{tabular}} (bot);
\draw[->] (1) to[loop, in=250, out=200, looseness=8] node[action, left, xshift=1mm, anchor=east, yshift=1mm] {$\{\cur{x}\,{\geq}\,\pre{x},\:\cur{x}\,{\neq}\,2\}$} (1);
\draw[->] (2) to[loop right, looseness=10] node[action, right, anchor=west, xshift=-7mm, yshift=4mm] {$\{\cur{x}\,{\geq}\,\pre{x}\}$} (2);
\draw[->] (2) -- node[action, right] {$\{ \cur{x}\,{\geq}\,\pre{x}\}$} (top);
\draw[->] (bot) to[loop, in=250, out=200, looseness=8]  (bot);
\draw[->] (2) -- node[action, below] {$\{\cur{x}\,{<}\,\pre{x}\}$} (bot);
\end{tikzpicture}
}
%\end{center}

\noindent
In the respective DFA $\DFA[\psi_\back]$ below, 
$\Theta_\curr$ consists of $\{\cur{x}\,{=}\,2\}$ and 
$\{\cur{x}\,{\neq}\,2\}$, and
$\Theta$ of all four combinations of
$\cur{x}\,{=}\,2$ or
$\cur{x}\,{\neq}\,2$ with $\cur{x}\,{\geq}\,\pre{x}$ or $\cur{x}\,{<}\,\pre{x}$:

\noindent
\resizebox{\columnwidth}{!}{
\begin{tikzpicture}[node distance=49mm]
 \node[state] (A) {$A$};
 \node[state, right of =A, final] (B) {$B$};
 \node[state, below of=A, yshift=25mm] (C) {$C$};
%  \node[state, below of=B, final] (D) {$D$};
 \node[state, right of=C] (E) {$D$};
\draw[edge] ($(A) + (-.4,0)$) -- (A);
\draw[edge] (A) -- node[action, below] {$\{\cur{x}\,{=}\,2\}$} (B);
\draw[edge] (A) -- node[action, left] {$\{\cur{x}\,{\neq}\,2\}$} (C);
\draw[->] (B) to[loop right, looseness=8]
  node[action, right, anchor=west] {$\{\cur{x}\,{\geq}\,\pre{x}, \cur{x}{=}2\}$}
  node[action, right, anchor=west, yshift=-5mm] {$\{\cur{x}\,{\geq}\,\pre{x}, \cur{x}{\neq}2\}$} (B);
\draw[edge] (B) -- 
  node[action, right, anchor=west] {$\{\cur{x}\,{<}\,\pre{x}, \cur{x}{=}2\}$}
  node[action, right, anchor=west, near end] {$\{\cur{x}\,{<}\,\pre{x}, \cur{x}{\neq}2\}$} (E);
\draw[->] (C) to[loop left, looseness=8] node[action, left, anchor=east] {$\begin{array}{r@{}l}\{\cur{x}\,{\geq}\,\pre{x},\\ \,\cur{x}\,{\neq}\,2\}\end{array}$} (C);
\draw[->] (C) to node[action, left, anchor=east, pos=.6] {$\begin{array}{r@{}l}\{\cur{x}\,&{\geq}\,\pre{x},\\\,\cur{x}\,&{=}\,2\}\end{array}$} (B);
\draw[->] (C) to 
  node[action, above, xshift=2mm] {$\{\cur{x}\,{<}\,\pre{x}, \cur{x}{=}2\}$}
  node[action, below, xshift=2mm] {$\{\cur{x}\,{<}\,\pre{x}, \cur{x}{\neq}2\}$} (E);
\draw[->] (E) to[loop right, looseness=8] node[action, right] {$\Theta$}  (E);
\end{tikzpicture}
}

\noindent
Here $A$ corresponds to $\{q_0\}$, 
$B$ to $\{q_2, q_+\}$ (equivalent to $\{q_2, q_+, \bot\}$),
% (in the subset construction this is merged with the equivalent state $\{q_2, q_+, \bot\}$),
$C$ to $\{q_1, \bot\}$, and
% $D$ to $\{q_2, q_+, \bot\}$, and
$D$ to $\{\bot\}$.
\end{example}

\section{Monitoring properties without lookahead}
\label{sec:nolookahead}

In this section we show that monitoring properties without lookahead is solvable, and
a monitoring structure is given by
the DFA $\DFA[\psi_\back] = (\QQ, \Theta, \delta, \{q_0\}, \QQ_F)$ (note that without lookahead, $\Theta_\curr$ and
$\Theta$ as defined in the last section coincide). We illustrate the construction on an example.

\begin{example}
\label{exa:DFA}
Let $\psi = (y\,{\geq}\,0) \U (x\,{>}\,y \wedge \G (x\,{>}\,y))$. 
Then $\psi_{\mathit{back}}$
% = (\curr{y}\,{\geq}\,0) \U (\curr{x}\,{>}\,\curr{y} \wedge \G (\curr{x}\,{>}\,\curr{y}))$,
is as $\psi$ but where $x$ is replaced by $\cur{x}$ and $y$ by $\cur{y}$. 
As there is no lookahead, no confusion can arise, so we write
$x$ for $\cur{x}$ and $y$ for $\cur{y}$.
We get the following NFA $\NFA[\psi_\back]$ with alphabet $\smash{\Sigma=2^{C^\pm}}$ for $C = \{x\,{>}\,y,\:y\,{\geq}\,0\}$:\\
% so $C\cup C^- = \{\curr{x}\,{>}\,\curr{y},\:\curr{x}\,{\leq}\,\curr{y},\:\curr{y}\,{\geq}\,0,\:\curr{y}\,{<}\,0\}$.
{\centering
\resizebox{.65\columnwidth}{!}{
\begin{tikzpicture}[node distance=35mm]
 \node[state, minimum width=6mm] (1) {1};
 \node[state, right of=1, minimum width=6mm, final] (2) {2};
\draw[edge] ($(1) + (-.3,.2)$) -- (1);
 \draw[edge] (1) -- node[action, above] {$\{x\,{>}\,y\}$} (2);
\draw[->] (1) to[loop left, looseness=8] node[action, left, anchor=east] {$\{y\,{\geq}\,0\}$} (1);
\draw[->] (2) to[loop right, looseness=8] node[action, right, anchor=west] {$\{x\,{>}\,y\}$} (2);
\end{tikzpicture}
}
}\\
The next DFA $\DFA[\psi_\back]$ is equivalent; its alphabet consists of all maximal satisfiable subsets of $C^\pm$,
i.e., 
$\{y{\geq}0,\,x{>}y\}$,
$\{y{\geq}0,\,x{\leq}y\}$,
$\{y{<}0,\,x{>}y\}$, and
$\{y{<}0,\,x{\leq}y\}$. \\
% The node coloring will be explained later.\\
{\centering
\resizebox{\columnwidth}{!}
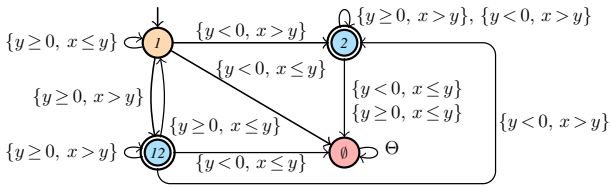
{
\begin{tikzpicture}[node distance=35mm]
 \node[state, minimum width=7mm, cvcolor, scale=.8] (1) {1};
 \node[state, right of=1, minimum width=7mm, cscolor, final, scale=.8] (2) {2};
 \node[state, below of=1, minimum width=7mm, cscolor, final, scale=.8, yshift=18mm] (12) {12};
 \node[state, right of=12, minimum width=7mm, pvcolor, scale=.8] (bot) {$\emptyset$};
\draw[edge] ($(1) + (0,.4)$) -- (1);
 \draw[edge] (1) -- node[action, above, yshift=-.9mm] {$\{y\,{<}\,0,\,x\,{>}\,y\}$} (2);
 \draw[edge, bend right=10] (1) to node[action, left, anchor=east] {$\{y\,{\geq}\,0,\,x\,{>}\,y\}$} (12);
 \draw[edge] (1) -- node[action, right, pos=.2, xshift=1mm] {$\{y\,{<}\,0,\,x\,{\leq}\,y\}$} (bot);
\draw[->] (1) to[loop left, looseness=8] node[action, left, anchor=east] {$\{y\,{\geq}\,0,\:x\,{\leq}\,y\}$} (1);
\draw[->] (12) to[loop left, looseness=8] node[action, left, anchor=east] {$\{y\,{\geq}\,0,\:x\,{>}\,y\}$} (12);
\draw[->, bend right=10] (12) to node[action, right, very near start] {$\{y\,{\geq}\,0,\:x\,{\leq}\,y\}$} (1);
\draw[->] (2) to[loop above, looseness=8, pos=.8]
 node[action, right] {$\{y\,{\geq}\,0,\:x\,{>}\,y\},\,\{y\,{<}\,0,\,x\,{>}\,y\}$} (2);
\draw[->] (bot) to[loop right, looseness=8] node[action, right, anchor=west, near start, xshift=1mm] {$\Theta$}  (bot);
\draw[->] (2) --
 node[action, right, anchor=west, pos=.7] {$\{y\,{\geq}\,0,\,x\,{\leq}\,y\}$} 
 node[action, right, anchor=west, pos=.4] {$\{y\,{<}\,0,\,x\,{\leq}\,y\}$}  (bot);
\draw[->] (12) -- node[action, below, yshift=.9mm] {$\{y\,{<}\,0,\,x\,{\leq}\,y\}$} (bot);
\draw[->, rounded corners] (12) --  ($(12) + (0,-.35)$) -|  ($(2) + (1.7,0)$) -- node[action, right, xshift=12mm, yshift=-15mm] {$\{y\,{<}\,0,\,x\,{>}\,y\}$} (2);
\end{tikzpicture}
}
}
\end{example}

Next, we formalize how to use $\DFA[\psi_\back]$ as a monitor. 
We call a state $ P'$ \emph{reachable} from $P$ if $P \to^*_w P'$ for some $w$.
% Below, by \emph{reachable} we mean reachable in 0 or more steps.

\begin{theorem}
\label{thm:lookahead0:monitoring}
If a word $w \inn \Theta^+$ is consistent with a trace $\tau$ and $P$ is the state
in $\DFA[\psi_\back]$ such that
$\{q_0\} \gotos{w} P$,
{
\setlength{\parindent}{0em}
\setdefaultleftmargin{0em}{2em}{}{}{}{}
\begin{compactitem}[]
\item $\tau \models \llbracket \psi\,{=}\,\PS\rrbracket$
if only final states are reachable from $P$,
\item $\tau \models \llbracket \psi\,{=}\,\CS\rrbracket$
if $P\inn \QQ_F$ but $P \to^* P'$ for some $P'\not\in \QQ_F$,
\item  $\tau \models \llbracket \psi\,{=}\,\CV\rrbracket$
if $P\,{\not\in}\, \QQ_F$ but $P \to^* P'$ for some $P'\in \QQ_F$,
\item  $\tau \models \llbracket \psi\,{=}\,\PV\rrbracket$
if no final state is reachable from $P$.
\end{compactitem}
}
\end{theorem}
\begin{proof}
% Note that by determinism, the word $w$ and state $P$ are uniquely defined.
By Lems.~\ref{lem:DFA:NFA} and
\ref{lem:automaton:acceptance}, $\tau \models \psi$ iff $P$ is final.
(Case $\PS$:) $P$ has no path to a non-final state.
Towards contradiction, assume $\tau'$ exists s.t. $\tau\tau' \not\models \psi$. 
By Lems.~\ref{lem:automaton:acceptance} and \ref{lem:DFA:NFA}, $\DFA[\psi_\back]$ does not accept a word consistent with $\tau\tau'$.
By Lem.~\ref{lem:word:for:trace},
there is a unique word $\bar w \in \Theta^+$ consistent with $\tau\tau'$
s.t. $\{q_0\} \to^*_{\bar w} P'$ in $\DFA[\psi_\back]$ for some $P'$, which cannot be accepting.
By determinism,
we can write $\bar w = ww'$ s.t. $\{q_0\} \to^*_w P \to^*_{w'} P'$, contradicting that only final states are reachable from $P$.
Thus $\tau'$ cannot exist. 
(Case $\CS$:) let $P \to^*_{w'} P'$ for a non-final state $P'$.
Sets in $\Theta$ are satisfiable, so there is a trace $\tau'$ s.t. $w'$ is consistent with $\tau'$, and $ww'$ with $\tau\tau'$.
By Lem.~\ref{lem:word:for:trace}, $ww'$ is the unique word in $\Theta^*$ consistent with $\tau\tau'$, and is not accepted. By Lems.~\ref{lem:DFA:NFA} and~\ref{lem:automaton:acceptance}, $\tau\tau' \not \models \psi$ thus $\tau \models \llbracket \psi = \CS\rrbracket$.
The other two cases are similar.
\end{proof}
% 
% In a similar way, we can show the following:
% 
% \begin{theorem}
% A property $\psi$ without lookahead is satisfiable iff $\DFA$ has a final state,
% and valid if all states except for the initial state in $\DFA$ are final.
% \end{theorem}

\noindent
Thm.~\ref{thm:lookahead0:monitoring} implies that monitoring is solvable,
and the monitoring state only depends on the DFA state matching a trace.
% This can be illustrated by re-considering Ex.~\ref{exa:DFA}.

\begin{example}
Consider the property of Ex.~\ref{exa:DFA}.
The following example trace is interleaved with the state reached by the corresponding word in $\DFA[\psi_\back]$, and the monitoring state:\\[1ex]
$\begin{array}{@{\:}c@{\:}c@{\:}c@{\:}c@{\:}c@{\:}c@{\:}c@{\:}c@{\:}c@{\:}c@{\:}c@{\:}c@{\:}}
\begin{array}{@{}c@{}} 1\\ \phantom{\CV} \end{array}
& \domino{x}{0}{y}{0} &
\begin{array}{@{}c@{}} 1\\ \CV \end{array} 
& \domino{x}{0}{y}{3} &
\begin{array}{@{}c@{}} 1\\ \CV \end{array} & \domino{x}{4}{y}{3} &
\begin{array}{@{}c@{}} 12\\ \CS \end{array} & \domino{x}{0}{y}{3} &
\begin{array}{@{}c@{}} 1\\ \CV \end{array} & \domino{x}{0}{y}{-1} &
\begin{array}{@{}c@{}} 2\\ \CS \end{array}
\end{array}$
% \\[1ex]
% $\begin{array}{@{\:}@{\:}c@{\:}c@{\:}c@{\:}c@{\:}c@{\:}c@{\:}c@{\:}c@{\:}c@{\:}c@{\:}c@{\:}c@{\:}}
% \begin{array}{@{}c@{}} 1\\ \CV \end{array} & \domino{x}{0}{y}{0} &
% \begin{array}{@{}c@{}} 1\\ \CV \end{array} & \domino{x}{-1}{y}{-1} &
% \begin{array}{@{}c@{}} \emptyset\\ \PV \end{array}
% \end{array}$
\\
%This illustrates that 
Each DFA state corresponds to a monitoring state,
as indicated by the colors in Ex.~\ref{exa:DFA}: non-final states are $\CV$ if a final state is reachable (yellow), $\PV$ otherwise (red); final states are $\CS$ if a non-final state is reachable (blue), $\PS$ otherwise. %(would be blue).
\end{example}

\section{Solvable cases for properties with lookahead}

In this section we study the more involved case of properties with lookahead 1.
The next example shows that the approach from the last section does not work for this case:

\begin{example}
\label{exa:lookahead:complicates:things}
% For 
% $\tau =
% \langle
% \{x\,{=}\,0\}, \{x\,{=}\,1\}, \{x\,{=}\,3\}, \{x\,{=}\,4\}
% \rangle$
Consider the property and DFA from Exa.~\ref{exa:lookahead} and the following trace,
interleaved with the reached state in $\DFA[\psi_\back]$ as well as the monitoring state for every  prefix:\\[1ex]
$\begin{array}{@{\:}c@{\:}c@{\:}c@{\:}c@{\:}c@{\:}c@{\:}c@{\:}c@{\:}c@{\:}c@{\:}}
\begin{array}{@{}c@{}} A\\ \phantom{\CV} \end{array}
& \{x=0\} &
\begin{array}{@{}c@{}} C\\ \CV \end{array} 
& \{x=1\}&
\begin{array}{@{}c@{}} C\\ \CV \end{array} 
& \{x=3\} &
\begin{array}{@{}c@{}} C\\ \PV \end{array} 
& \{x=4\} &
\begin{array}{@{}c@{}} C\\ \PV \end{array} 
\end{array}$\\
The reached DFA state is the same for every prefix, even though the monitoring state changes. 
% This shows that $\DFA[\psi_\back]$ does not serve as a monitoring structure.
\end{example}

Note that while the state $P$ reached by a trace $\tau$ in $\DFA[\psi_\back]$ can still tell whether $\tau$ satisfies $\psi$, reachability of (non)final states from $P$
is no longer sufficient to predict whether the trace can be \emph{extended} to satisfy $\psi$.
We therefore develop an enhanced monitoring structure.
To simplify notation, we assume that $\psi$ is the result of the $\back(\cdot)$ transformation, i.e., a property over variables $V_\prev \cup V_\curr$.
Moreover, we assume throughout this section that $\DFA$ is a DFA for $\psi$
of the form $\DFA= (\QQ, \Theta\,{\cup}\,\Theta_\curr, \Delta, \{q_0\}, \QQ_F)$.

% The following definitions apply the approach from~\cite{FMW22a} to automata for formulas rather than given transition systems.
The idea is to abstract the variable configurations that are relevant for verification by a set of formulas.
This abstraction requires some additional notation:
First, let $V_0 = \{v_0 \mid v\inn V\}$ be a copy of $V$, it will act
as placeholders for the initial values of the variables $V$.
We write $\Cinit$
for the formula $\Cinit = \bigwedge_{v\in V}v = v_0$.
% The following notions will be used to capture this abstraction.
For a formula $\varphi$ with free variables $V_0 \cup V$, let $\varphi(\vec U)$ denote the formula
where $\vec V$ is replaced by $\vec U$, but $V_0$ is not replaced (with denote by $\vec V$ a fixed ordering of $V$).
For a set of constraints $C$ with free variables $V_\prev \cup V_\curr$, 
let $C(\vec U, \vec V)$
denote the constraints where $\vec V_\prev$ is replaced by $\vec U$ and $\vec V_\curr$ is replaced by $\vec V$.

\begin{definition}
\label{def:update}
For a formula $\varphi$ with free variables $V_0 \cup V$
and a set of constraints $C$ with free variables $V_\prev \cup V_\curr$,
let 
$\smash{\update(\varphi, C) = \exists \vec U. \varphi(\vec U) \wedge \bigwedge C(\vec U, \vec V)}$.
\end{definition}

%As $\varphi$ has free variables $V_0 \cup V$, also $\update(\varphi, C)$ has free variables $V_0 \cup V$. 
Both $\varphi$ and $\update(\varphi, C)$ have free variables $V_0 \cup V$. 
Although $\update$ introduces existential quantifiers, there is an equivalent quantifier-free formula for $\update(\varphi, C)$: linear arithmetic has quantifier elimination.
E.g., for  $\varphi = (x\,{=}\,x_0 \wedge x_0\,{\neq}\,2)$ and
$C = \{\cur{x}\,{\geq}\,\pre{x}, \cur{x}\,{=}\,2\}$, we have
$\update(\varphi, C) = \exists u.\: u\,{=}\,x_0 \wedge x_0\,{\neq}\,2 \wedge
x\,{\geq}\,u \wedge x\,{=}\,2$, which is equivalent to $x_0 < 2 \wedge x=2$.

In order to design a monitoring structure,
we build a \emph{constraint graph} which, intuitively,
connects in a graph all variable dependencies 
(described by formulas) emerging from a state $P_0$ in the DFA, and where ${V_0}$ acts as placeholder for the initial variable values. 
The constraint graph is parameterized by an equivalence relation $\sim$ on formulas 
%that can help to yield a more compact presentation. 
used to reduce its size. 
Our correctness proof will require further properties of $\sim$
but a default choice is the logical equivalence relation $\equiv$.

\begin{definition}\label{def:cg}
A \emph{constraint graph} $\CG_\psi(P_0, \sim)$ for $\DFA$ and a DFA state $P_0 \inn \QQ$
% and an equivalence relation $\sim$ on formulas 
is a triple 
$\langle S, s_0, \gamma, S_F\rangle$ where the node set $S$
consists of tuples $(P, \varphi)$ of a state $P\inn \QQ$ and a formula $\varphi$ with free variables $V_0 \cup V$,
and $\gamma \subseteq S \times \Sigma \times S$. Then, $S$ and $\gamma$ are inductively defined as follows:
\begin{compactitem}
\item[$(i)$] $s_0 = (P_0, \Cinit)$ is the initial node and $s_0\in S$;
\item[$(ii)$] if $s = (P,\varphi) \in S$ and $P \goto{\varsigma} P'$ in $\DFA$ such that
$\update(\varphi, \varsigma)$ is satisfiable,
there is some $s' = (P', \varphi')\in S$
with $\varphi' \sim \update(\varphi, \varsigma)$, and 
$s \goto{\varsigma} s'$ is in $\gamma$;
\item[$(iii)$]  $(P, \varphi)\in S$ is in the set of final nodes $S_F$ iff $P\inn \QQ_F$.
\end{compactitem}
\end{definition}

We simply write $\CG_\psi(P_0)$ for $\CG_\psi(P_0, \equiv)$. 
The constraint graph $\CG_\psi(C)$ for the DFA $\DFA[\psi_\back]$ from
Ex.~\ref{exa:lookahead} is shown in Fig.~\ref{fig:cg}.
% \todo{Marco: The second part of the sentence was commented. I uncommented it\ldots}
For readability, edge labels are combined; e.g. the two
edges from $B$ to $D$ labeled
$\{\cur{x}\,{<}\,\pre{x}, \cur{x}\,{=}\,2\}$ and $\{\cur{x}\,{<}\,\pre{x}, \cur{x}\,{\neq}\,2\}$ are combined to one edge labeled $\{\cur{x}\,{<}\,\pre{x}\}$.
\longversion{The graphs $\CG_\psi(A)$ and $\CG_\psi(B)$ are shown in Ex.~\ref{exa:cgs} in the appendix.}{}
\begin{figure}
\begin{tikzpicture}[node distance = 17mm, ->,>=stealth',shorten >=1pt]
\node[node, below of=0, xshift=-55mm] (2)
 {\cgnode{$C$}{$x\,{=}\,x_0$}};
\node[node, below of=2, xshift=-35mm] (21)
 {\cgnode{$C$}{$x \geq x_0 \wedge x\,{\neq}\,2$}};
\node[node, below of=2, final] (22)
 {\cgnode{$B$}{$x_0\,{\leq}\,2 \wedge x\,{=}\,2$}};
\node[node, right of=22, xshift=16mm] (23)
 {\cgnode{$D$}{$x\,{<}\,x_0$}};
\node[node, below of=23] (231)
 {\cgnode{$D$}{$\top$}};
\node[node, below of=22] (222)
 {\cgnode{$D$}{$x_0\,{\leq}\,2$}};
\node[node, below of=21, final, xshift=-7mm] (221)
 {\cgnode{$B$}{$x_0\,{\leq}\,2 \wedge x \geq 2$}};
\draw[goto, bend right=15] (2) to node[action, above, anchor=south east, near start]{$\{\cur{x}{\geq}\pre{x}\wedge \cur{x}{\neq}2\}$} (21);
\draw[goto] (2) to node[action, right, anchor= west]{\begin{tabular}{@{}c@{}}$\{\cur{x}{\geq}\pre{x}$\\${}\wedge \cur{x}{=}2\}$\end{tabular}} (22);
\draw[goto, bend left=15] (2) to node[action, above, near start, anchor=south west]{$\{\cur{x}{<}\pre{x}\}$} (23);
\draw[goto, loop above, out=95, in=66, looseness=10] (23) to  node[action, above, xshift=4mm]{$\{\cur{x}{<}\pre{x}\}$} (23);
\draw[goto] (23) to  node[action, right]{$\{\cur{x}{\geq}\pre{x}\}$} (231);
\draw[goto] (22) to node[action, above, anchor=south west, near end]{$\{\cur{x}{<}\pre{x}\}$} (222);
\draw[goto] (22) to  node[action, right, near end, xshift=2mm]{$\{\cur{x}{\geq}\pre{x}\}$} (221.100);
\draw[goto, loop right, out=175, in=185, looseness=7] (221) to  node[action, above, xshift=-4mm, yshift=1mm]{$\{\cur{x}{\geq}\pre{x}\}$} (221);
\draw[goto] (221) to node[action, below]{$\{\cur{x}{<}\pre{x}\}$} (222);
\draw[goto, loop left, out=150, in=170, looseness=3] (21) to node[action, above]{$\{\cur{x}{\geq}\pre{x}\wedge \cur{x}{\neq}2\}$} (21);
\draw[goto, bend right=20, in=191] (21) to node[action, below, anchor=north east, near start, xshift=-1mm, yshift=1mm]{$\{\cur{x}{\geq}\pre{x}\wedge \cur{x}{=}2\}$} (22);
\draw[goto, loop right, out=13, in=-5, looseness=7] (222) to  node[action, right]{$\widehat \Sigma$} (222);
\end{tikzpicture}
\caption{\label{fig:cg}Constraint graph $\CG_\psi(C)$.}
\end{figure}
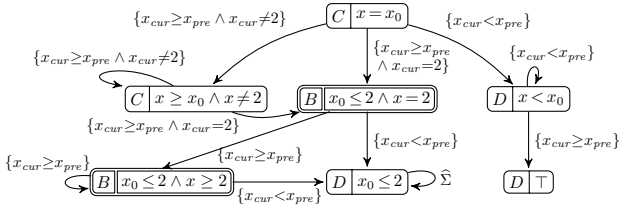
The construction in Def.~\ref{def:cg} need not terminate as infinitely many formulas may occur, but we will show in the next section that termination is guaranteed for  several relevant cases. 
%For these cases, we define the formulas:%, having free variables $V_0$:
\begin{definition}
\label{def:future}
For $\psi$ with DFA $\DFA$, a state $P_0$ in $\DFA$, and 
 $G:= \CG_\psi(P_0,\sim)$ with node set $S$, let\\[.5ex]
$\begin{array}{r@{\,}l}
\sat(G) =&
\exists \vec V. \bigvee \{\varphi \mid (P, \varphi)\in S \text{ is a final node} \}
\\
\unsat(G) =&
\exists \vec V. \bigvee \{\varphi \mid (P, \varphi) \in S \text{ is a non-final node} \}
\end{array}$
\end{definition}

Intuitively, $\sat(G)$ expresses a condition on the variable values when a final state in $\DFA$ is reachable from $P_0$; dually, $\unsat$ gives a condition for reachability of non-final states.
This leads to the following monitoring procedure:
%\begin{algorithm}
%\caption*{}\label{alg:monitoring}
\begin{algorithmic}[1]
\Procedure{monitor}{$\psi$, $\tau$}
  \State compute $\psi_\back$ and $\DFA[\psi_\back]$ \Comment{or  take from cache}
  \State $w \gets$ word in $\Theta_\curr\Theta^*$ consistent with $\tau$
  \State $P \gets$ state in $\DFA[\psi_\back]$ such that $\{q_0\}\gotos{w} P$
  \State $G\gets\CG_\psi(P, \sim)$ \Comment{or take from cache}
  \State $\alpha_n \gets $ last assignment in $\tau$
  \If{$P$ accepting in $\DFA[\psi_\back]$}
    \State \textbf{return} ($\CS$ \textbf{if} $\alpha_n \models\unsat(G)$ \textbf{else} $\PS$)
  \EndIf
  \State \textbf{else return} ($\CV$ \textbf{if} $\alpha_n \models\sat(G)$ \textbf{else} $\PV$)
\EndProcedure
\end{algorithmic}
%\end{algorithm}

When monitoring multiple traces against a property $\psi$,
the mo\-ni\-toring structures (DFA, constraint graphs for every state, $\sat$, $\unsat$) can be computed once and for all upfront. Instead, the procedure shown above is intended as a \emph{lazy} implementation, where the structures are computed when needed and cached. 
%Alternatively, having built $\DFA[\psi_\back]$, one could precompute $G$ for all its states beforehand. 
% the DFA and the constraint graphs need to be computed only once. 
%Before proving correctness, 
%We apply the procedure to an example:

\begin{example}
For
$\psi_{\mathit{back}} = \G\wX (\cur{x}\,{\geq}\,\pre{x}) \wedge \F (\cur{x}\,{=}\,2)$
from Exa.~\ref{exa:lookahead}, by the graph in Fig.~\ref{fig:cg} we have
$\sat(C) = \exists x. (x_0\,{\leq}\,2 \wedge x\,{\geq}\,2) \vee (x_0\,{\leq}\,2 \wedge x\,{=}\,2)$, equivalent to $x_0 \leq 2$. We show how the procedure distinguishes two prefixes of the trace from Exa.~\ref{exa:lookahead:complicates:things} with different monitoring state:
\begin{compactitem}
% \item
% It can be readily checked that the empty trace $\tau_0$ does not satisfy $\psi$.
% As $\CG_\psi(A)$ has a final state, $\tau_0 \models \llbracket \psi{=}\CV \rrbracket$.
\item
Trace $\tau_1 = \langle\{x\,{=}\,0\}\rangle$ corresponds to $w = \langle\{\cur{x}\,{\neq}\,2\}\rangle$, leading to the non-accepting state $C$.
As $\{x\,{=}\,0\} \models \sat(G_C)(V) = x\leq 2$, we have
$\tau_1 \models \llbracket \psi{=}\CV \rrbracket$.
% \item
% For $\tau_2=\langle\{x\,{=}\,0\}, \{x\,{=}\,1\}\rangle$, the case is similar,
% for $w = \langle\{\cur{x}\,{\neq}\,2\}, \{\cur{x}\,{\neq}\,2, \cur{x}{\geq}\pre{x}\}\rangle$.
\item
Trace $\tau_3 =\langle\{x\,{=}\,0\}, \{x\,{=}\,1\}, \{x\,{=}\,3\}\rangle$ matches
$\langle\{\cur{x}\,{\neq}\,2\},$ $\{\cur{x}\,{\neq}\,2, \cur{x}{\geq}\pre{x}\}^2\rangle$, leading again to state $C$.
However, 
$\{x\,{=}\,3\} \not\models x\leq 2$, so
$\tau_3 \models \llbracket \psi{=}\PV \rrbracket$.
\end{compactitem}
% The cases of the others prefixes are similar.
\end{example}

\paragraph{Correctness.}
To prove correctness,
we need some additional notions. First, we define \emph{history constraints}
to capture the formulas obtained from stacked $\update$
operations:

\begin{definition}\label{def:history constraint}
For $w\,{=}\,\varsigma_0,\dots, \varsigma_{n-1}$ in $\Theta^*$,
the \emph{history constraint} $\hist(w)$ is given by
$\hist(w) = \Cinit$ if $n= 0$, and if $n\,{>}\,0$ then
$\hist(w) = \update(\hist(\langle\varsigma_0,\dots, \varsigma_{n-2}\rangle), \varsigma_{n-1})$.
\end{definition}
%History constraints are formulas with free variables $V\cup V_0$, 
%The crucial property of history constraints is that 
%History constraints are so that 
There is the correspondence between satisfying assignments of
$h(w)$ and traces that satisfy 
%the verification obligations
%given by the constraints in $w$, as stated next:
the constraints in $w$: 

\begin{restatable}{lemma}{lemmaabstraction}
\label{lem:abstraction}
For $w\in \Theta^*$ of length $n$,
$\hist(w)$ is satisfied by assignment $\nu$ iff
$\langle\emptyset\rangle\,{\cdot}\,w$ is consistent with a trace $\alpha_0, \dots, \alpha_{n}$ 
such that $\nu(v_0)\,{=}\, \alpha_0(v)$ and $\nu(v)\,{=}\, \alpha_n(v)$ for all $v\in V$.
\end{restatable}

Next, we define a \emph{history set} to capture all variable configurations that are relevant for monitoring. More precisely, it collects pairs of history constraints for words and the state that they lead to, but where equivalent pairs are eliminated:

\begin{definition}
A \emph{history set} for $\DFA$ is a minimal set $\Phi$ of pairs $(P, \varphi)$
of a state $P \inn \QQ$ and a formula $\varphi$ such that for all $P'$ and $w\in \Theta_\curr\Theta^*$ with $P \gotos{w} P'$ in $\DFA$, there is a $(P',\varphi) \in \Phi$ s.t.
$\hist(w) \equiv \varphi$.
% ; and $\Phi$ contains no other pairs.
\end{definition}

% We will show later that all nodes in a CG can be taken from a history set.
% Intuitively, a history set has a representative for every variable configuration that is relevant when monitoring $\psi$.
%If a history set if finite, it serves as an abstract finite state representation of the relevant variable configurations, and we show below that nodes in a CG can be taken from a history set.
If $\Phi$ is finite, it represents an abstract representation of the relevant variable configurations and can be used to build the states of the constraint graph. 
%, so we show that nodes in a CG are taken from a history set.
%
For succinctness, we pair a history set with an equivalence relation. %and call this a \emph{summary}:
\begin{definition}
\label{def:summary}
A \emph{summary} for $\DFA$
is a pair $(\Phi, \sim)$ of a history set $\Phi$ for $\DFA$
and equivalence relation $\sim$ on $\Phi$ such that
\begin{compactenum}[(1)]
\item $\sim$ contains $\equiv$ on formulas in $\Phi$ and is decidable,
\item for all $(q,\varphi),(q,\psi) \in \Phi$ such that  $\varphi \sim \psi$,
\begin{inparaenum}
\item  if $\alpha \models \varphi$, there is some $\alpha'$ such that $\alpha' \models \psi$ and $\alpha(u) = \alpha'(u)$ for all $u \in V_0$, and 
% \item $\varphi$ and $\psi$ are equisatisfiable, and
\item for all 
transitions $q \goto{\varsigma} q'$,
$\update(\varphi,\varsigma) \sim \update(\psi,\varsigma)$.
\end{inparaenum}
\end{compactenum}
\end{definition}

The summary $(\Phi, {\sim})$ is \emph{finite} if
$\sim$ has finitely many equivalence classes, and
$\psi$ \emph{has finite summary} if a finite summary exists for $\DFA$.
For a word $w\,{=}\, \varsigma_1, \dots, \varsigma_n$, we write $(P_0, \Cinit) \gotos{w} (P, \varphi)$
if there is a path
$(P_0, \Cinit) \goto{\varsigma_1} \dots \dots \goto{\varsigma_n} (P, \varphi)$
in $\CG_\psi(P_0, \sim)$.
The next lemma states that a path in the CG emulates a path in $\DFA$ and leads to a state that is equivalent to a history constraint. 

\begin{restatable}{lemma}{lemmacg}
\label{lem:cg}
Let $(\Phi, \sim)$ be a summary for $\DFA$.
\begin{inparaenum}[(a)]
\item If $\CG_\psi(P_0, \sim)$ has a path $(P_0, \Cinit) \gotos{w} (P, \varphi)$ then
$P_0 \gotos{w} P$ in $\DFA$ and $\varphi \sim \hist(w)$ is satisfiable, and
\item if $P_0 \gotos{w} P$ and $h(w)$ is satisfiable then there is a path
$(P_0, \Cinit) \gotos{w} (P, \varphi)$ for some $\varphi$ such that $\varphi \sim \hist(w)$.
\end{inparaenum}
\end{restatable}

\noindent
At this point we are ready to prove our main theorem:

\begin{theorem}
\label{thm:lookahead1:monitoring}
Let $(\Phi, \sim)$ be a summary for $\DFA$.
Given a DFA $\DFA$ for $\psi$ and a trace $\tau$, let $w\in \Theta_\curr\Theta^*$ be
the word consistent with $\tau$ and
$P$ the $\DFA$ state
s.t. $\{q_0\} \gotos{w} P$. For $G:= \CG_\psi(P, \sim)$ the constraint graph from $P$,
\begin{compactitem}
\item
if $P\,{\in}\,P_F$ then $\tau \models \llbracket \psi{=}\CS \rrbracket$ if
$\alpha_n \models \unsat(G)(\vec V)$, and $\tau \models \llbracket \psi{=}\PS \rrbracket$ otherwise,
\item
if $P\,{\not\in}\,P_F$ then $\tau \models \llbracket \psi{=}\CV \rrbracket$ if
$\alpha_n \models \sat(G)(\vec V)$, and $\tau \models \llbracket \psi{=}\PV \rrbracket$ otherwise.
\end{compactitem}
\end{theorem}

\begin{proof}[Proof]
Note that the word $w$ that is
consistent with $\tau$ is unique (cf. Lem.~\ref{lem:word:for:trace}).
Consider the case where $P\,{\in}\,P_F$, so that $\tau \models \psi$ by Lem. \ref{lem:automaton:acceptance} and \ref{lem:DFA:NFA}. 
First, suppose $\alpha_n \models \unsat(G)(\vec V)$, so
$\alpha_n \models (\exists \vec V.\,\varphi)(\vec V)$
for some non-final node $(P', \varphi)$ of $G$.
Let $u$ be a word s.t. there is a path $(P, \Cinit) \gotos{u} (P', \varphi)$ in
$\CG_\psi(P)$.
By Lem. \ref{lem:cg}, $P \gotos{u} P'$ in $\DFA$ and $\varphi \sim \hist(u)$.
By Def.~\ref{def:summary}, $\alpha_n \models (\exists \vec V.\,\hist(u))(\vec V)$, so
there is an assignment $\nu$ with domain 
$V_0 \cup V$ s.t. $\nu(v_0) = \alpha_n(v)$ for all $v\inn V$ and 
$\nu \models h(u)$.
By Lem. \ref{lem:abstraction}, $\langle \emptyset\rangle\,{\cdot}\,u$ is consistent with some trace $\tau' = \langle \alpha_0', \dots, \alpha_k'\rangle$ s.t. $\nu(v_0) = \alpha_0'(v)$ for all $v\in V$.
Thus, $\alpha_n$ and $\alpha_0'$ coincide, and $wu$ is consistent with
the trace $\tau'' = \langle \alpha_0, \dots, \alpha_n, \alpha_1', \dots, \alpha_k'\rangle$
as $w$ is consistent with $\tau$ and $u$ with $\tau'$.
As $\{q_0\} \gotos{wu} P'$ is not accepting, by Lem. \ref{lem:automaton:acceptance}, $\tau'' \not\models \psi$, hence $\tau \models \llbracket \psi{=}\CS \rrbracket$.

Second, let $\alpha_n \not\models \unsat(G)(\vec V)$.
Towards a contradiction, suppose there is a trace $\tau' = \langle \alpha_1', \dots, \alpha_k'\rangle$ s.t.
$\tau\tau' \not\models \psi$.
Using Lem.~\ref{lem:word:for:trace}, there is a word $u\in\Theta^*$ s.t. $wu$ is consistent with $\tau\tau'$; let $P'$ the state s.t. $\{q_0\} \gotos{wu} P'$.
By Lem.~ \ref{lem:automaton:acceptance} and \ref{lem:DFA:NFA}, the state $P'$ is not final.
As $wu$ is consistent with $\tau\tau'$, $\langle \emptyset\rangle \cdot u$ is consistent with $\langle\alpha_n\rangle\tau'$, so
by Lem. \ref{lem:abstraction}, the assignment $\nu$ s.t.
$\nu(v_0) = \alpha_n(v)$ and $\nu(v) = \alpha_k'(v)$ for all $v\in V$ satisfies 
$\hist(u)$.
Thus Lem. \ref{lem:cg} implies that $\CG_\psi(P)$ has a path $(P, \Cinit) \gotos{u} (P', \varphi)$ s.t. $\varphi \sim \hist(u)$.
As $\nu \models\hist(u)$, by Def.~\ref{def:summary}, $\alpha_n \models (\exists.\,\varphi)(\vec V)$ because $\nu(v_0) = \alpha_n(v)$ for all $v\in V$.
As $P'$ is non-final, it follows that $\alpha_n \models \unsat(G)$, a contradiction; so $\tau'$ cannot exist.
Thus $\tau \models \llbracket \psi{=}\PS \rrbracket$.
The case where $\tau \not\models \psi$ is dual, using $\sat(G)$.
\end{proof}

% A respective statement for the empty trace is the following; it is proven in a similar way.
% 
% \begin{theorem}
% \label{thm:lookahead1:sat}
% Let $\psi$ be a finite history property with DFA $\DFA$ whose initial state is $\{q_0\}$.
% Then there is a trace $\tau$ such that $\tau \models \psi$ iff $\CG_\psi(\{q_0\})$ has a
% final state, and there is a trace $\tau$ such that $\tau \not\models \psi$ iff
% $\CG_\psi(\{q_0\})$ has a non-final state.
% % equivalent formulation:
% % $\sat(\{q_0\})$ is satisfiable, and there is a trace $\tau$ such that $\tau \not\models \psi$ iff $\unsat(\{q_0\})$ is satisfiable.
% \end{theorem}

Finally, if $\psi$ has a finite summary $(\Phi,\sim)$, by~Lem. \ref{lem:cg} a constraint graph $\CG(P,\sim)$ computed from $\DFA$ can take all nodes from the finitely many equivalence classes of
$\Phi$, so that the constraint graph can be finite. Therefore:

\begin{corollary}
\label{cor:main}
Monitoring is solvable for $\psi\in\LL$  if $\psi$ has finite summary.
\end{corollary}

\section{Concrete criteria for solvability}

We next apply Cor.~\ref{cor:main} to show
solvability of concrete property classes, by restricting the constraint set or control flow.

\paragraph{Monotonicity constraints} (MCs) over variables $\mathcal V$ and domain $D$ have the form $p \odot q$ where $p,q\in {D\,{\cup}\,\mathcal V}$
and $\odot$ is one of $=, \neq, \leq$, or $<$.
We call an LTL$_f$ property whose constraint atoms are MCs over $D$ an \emph{MC$_D$ property}.
Ex.~\ref{exa:lookahead} gives a simple example of an MC property. However,
MC properties are rich enough to capture practically important settings:
For instance, \citeauthor{GeistRS14} (\citeyear{GeistRS14}) model the specification of a fluxgate magnetometer of a Swift UAS system operated by NASA, as linear-time properties with arithmetic comparisons of sensor values, which are all MCs over $\mathbb Q$ (cf. their Table 2).
We have that:

\begin{restatable}{theorem}{theoremMC}
\label{thm:mc}
Monitoring is solvable for MC$_\mathbb Q$ properties.
\end{restatable}

\begin{proof}[Proof (sketch)]
This is shown as in~\cite[Thm.~5.2]{FMW22a}, exploiting that quantifier elimination of MC formulas over $\mathbb Q$ produces MC formulas with the same  constants, so that a finite history set exists.
\end{proof}

%This reasoning fails for MC$_{\mathbb Z}$ properties, but we will come back to this case when discussing gap-order constraints.

\paragraph{Integer periodicity constraints} (IPCs) confine the constraint language in a similar way as MCs, but allow equality modulo, and variable-to-variable comparisons are restricted. IPCs are e.g. used in calendar formalisms~\cite{Demri06}.
More precisely, IPC atoms have the form $x = y$, $x \odot d$ for $\odot \in \{=,\neq, <, >\}$, $x \equiv_k y + d$, or $x \equiv_k d$, for variables $x,y$ with domain $\mathbb Z$ and $k,d\in \mathbb N$. 
An $\LL$ property over IPC atoms is called an \emph{IPC property}.
The next result is proven similarly as Thm.~\ref{thm:mc}, cf.~\cite[Thm.~4]{FMW22c}.

\begin{theorem}
\label{thm:ipc}
Monitoring is solvable for IPC properties.
\end{theorem}

A simple example of an IPC property is $(x \equiv_7 y+1)\U (x\,{=}\,z)$.
Also e.g. a parallel program as in~\cite[Fig.~3]{HavelundRR19} can be modeled
as an IPC property, where monitoring detect at runtime whether a trace may trigger the race on variable $x$ that the program exhibits.

\paragraph{Control flow restrictions} can be used instead of confining the constraint language to obtain solvability.
\emph{Computation graphs}~\cite{DDV12} were introduced to track dependencies between variable instances at different trace events.
A property $\psi$ is said to have $k$-\emph{bounded lookback}~\cite{FMW22a} if all paths in computation graphs of $\psi$ have length at most $k$, without counting equality edges. Intuitively, this corresponds to forbid variable updates that depend on an unbounded history of values (of the same or of other variables), but only on a $k$-bounded ``moving window". 
A property has bounded lookback (BL) if its has $k$-bounded lookback for some $k\geq 0$.
E.g., $\G (x'\,{>}\,x)$ does not have BL, because there is an unbounded dependency chain between the values of $x$, whereas $\F (x'\,{>}\,2y) \wedge \G (x{+}y\,{>}\,0)$ has 1-bounded lookback as variable comparisons span at most one time unit. 
Also all properties with lookahead 0 have BL. 
% \todo{mention?}
%Bounded lookback properties turn out to be solvable for monitoring 
Monitoring of BL properties is solvable 
because, roughly, the quantifier depth of history constraints is upper-bounded a priori.
A formal definition, proof and examples are in \longversion{the appendix}{\cite{longversion}}.

\begin{restatable}{theorem}{theoremboundedlookback}
\label{thm:bounded:lookback}
Monitoring of BL properties is solvable.
\end{restatable}

\paragraph{Gap-order constraints} (GCs) have the form
$x - y \geq k$ for $x$ and $y$ either variables with domain $\mathbb Z$ or integers, and $k\inn \mathbb N$.
Satisfiability of GC properties is decidable~\cite{BP14}.
However, validity is undecidable, because a GC property $\psi$ is valid iff $\neg\psi$ is unsatisfiable, where $\neg \psi$
can be written as a property whose atoms are negated GCs, i.e. of the form $\neg c$ for $c$ a GC. Negated GCs can model counters~\cite[Thm. 12]{BP14}, so that state reachability of a 2-counter machine can be reduced to satisfiability of a property over negated GCs.
Consequently,

\begin{restatable}{remark}{remarkGC}
Monitoring of GC properties is not solvable.
\end{restatable}

MCs over $\mathbb Z$ can be written as GCs~\cite{BP14}, with the key difference that
for an MC $c$ also $\neg c$ is an MC, and hence a GC. Thus, we can assume that in a DFA for an $\text{MC}_{\mathbb Z}$ property $\psi$, all constraints are GCs.
% Cor.~\ref{cor:main} can be used to show that monitoring is solvable for such $\psi$, but this case is more involved than those above, as discussed next. 

In \cite[Thm.~5.5]{FMW22a} it was shown that a transition system over GCs admits a finite summary $(\text{GC}_K, \sim_K)$. Here, $\text{GC}_K$ is the finite set of quantifier-free formulas with GC atoms over variables $V$ and constants $\leq K$, and $\sim_K$ is the cutoff equivalence relation, where equivalence of two formulas is
checked after replacing constants greater than $K$ by $K$, for $K$ the maximal difference between constants in the input.
% $\varphi \sim_K \varphi'$ iff $\cutoff[K]{\varphi} \equiv \cutoff[K]{\varphi'}$, where
% $\cutoff[K]{\varphi}$ is the formula obtained from $\varphi$ by replacing all constants $k > K$ by $K$.
Finite summary relies on the initial assignment of $V$ being fixed, 
whereas by Def.~\ref{def:history constraint} history constraints start from $\Cinit$, i.e., a parametric assignment $\vec V \mapsto \vec V_0$. This can hamper finiteness of the summary. E.g., for $\psi=\X_w (\cur{x}{>}\pre{x} \U \cur{x}{>}5)$ there are finitely many equivalence classes only if the initial value $z$ of $x$ is known (after at most $5-z$ steps $\psi$ holds).
Nonetheless, Cor.~\ref{cor:main} can be used to show that monitoring is solvable, although the procedure needs to be suitably refined.

\begin{restatable}{theorem}{theoremMCZ}
\label{thm:mc}
Monitoring is solvable for MC$_\mathbb Z$ properties.
\end{restatable}

\begin{proof}[Proof (sketch)]
We modify our procedure as follows:
(1) in line 5, we compute the CG starting from 
$\bigwedge v{=}\alpha_n(v)$ instead of $\Cinit$.
By \cite[Thm.~5.5]{FMW22a} a finite summary exists, so the CG computation terminates.
(2) The checks in lines 8,9 are simplified: as the initial assignment is now integrated in the CG, at line 8 it suffices to check whether the CG has a non-empty path to a non-final state (resp. to a final state in line 9).
It is not hard to adapt the proof of Thm.~\ref{thm:lookahead1:monitoring} accordingly.
However, the changed procedure comes with the price of computing CGs
not only for every DFA state but also every assignment in the trace.
\end{proof}

\section{Conclusions}
\label{sec:conclusions}

\paragraph{Implementation.}
We implemented our approach in a prototype, whose source code and
web interface are available via \emph{https://bit.ly/3QFoJHA}.
The tool takes an \altlf property $\psi$ and a trace as input,
determines whether $\psi$ is in one of the 
decidable classes identified in the last section, constructs and visualizes a monitor, and computes the monitoring state.
The tool is implemented in Python, using Z3~\cite{Z3} and CVC5~\cite{DR0BT14} for SMT checks and quantifier elimination.

\paragraph{Future work.}
We want to lift our automata-based approach to the case of traces of richer states, equipped with full-fledged relations. Hence we plan to study how to integrate our approach with \cite{DeLT16} and \cite{CDMP22}. The former %\cite{DeLT16} 
considers states consisting of first-order structures modulo theories, but does not foresee any form of lookahead in the properties. In \cite{CDMP22}, instead, states are labeled by first-order interpretations and properties are expressed in a fragment of first-order LTL with a controlled first-order quantification across time; however arithmetic theories are not supported. None of these approaches feature an automata-based characterization of monitors, which makes the integration with the current work particularly interesting.

\bibliography{references}

\appendix
\longversion{\newpage\newpage
\section{Appendix}

\subsection{Preliminaries}
%\lemmaundecidable*

\begin{lemma}
\label{lem:undecidable}
Satisfiability of $\LL$ properties with lookahead 1 is undecidable.
%\todo{lookahead defined wrt variables. the formula $\psi$ below has no lookahead S: don't get it}
\end{lemma}

\begin{proof}
%For readability, we here write $v'$ in place of $v^1$.
%One can easily 
We model state reachability in a 2-counter machine $M$ as a property in $\LL$ with lookahead 1:
Let $V = \{x,y,s\}$ where $x,y$ model the counters, and $s$ is an additional variable keeping track of the control state. We fix for every state $q$ of $M$ a value $n_q \in \mathbb Z$; let $n_0$ be the value of $M$'s initial state.
For a transition $t \colon q \to p$ in $M$, let $\psi_t$ be the property 
\begin{inparaitem}
\item[$(i)$]
$s\,{=}\,n_q \wedge s^1\,{=}\,n_{p} \wedge x^1\,{=}\,x{+}1 \wedge y^1\,{=}\,y$ if 
$t$ increments $x$
\item[$(ii)$]
$s\,{=}\,n_q \wedge s^1\,{=}\,n_{p} \wedge x > 0 \wedge x^1\,{=}\,x{-}1 \wedge y^1\,{=}\,y$ if it decrements $x$
\item[$(iii)$]
$s\,{=}\,n_q \wedge s^1\,{=}\,n_{p} \wedge x\,{=}\,0 \wedge y^1\,{=}\,y\wedge x^1 = x$ if it checks $x$ for 0
\end{inparaitem}
and similar for $y$.
Then, state $f$ is reachable in $M$ iff $(s\,{=}\,n_0) \wedge \F (s\,{=}\,n_f) \wedge \G (\bigvee_t \psi_t)$ is satisfiable.
\end{proof}

\lemmamorelookahead*
\begin{proof}
We define a sequence of formulas $\psi_j$ of lookahead $j$, $0\,{\leq}\,j\,{\leq}\,m$, as follows:
$\psi_m = \psi$; and for $0\,{\leq}\,j\,{<}\,m$ we set $\psi_j = d(\psi_{j+1})$, where $d(\psi_{j+1}) = \psi_{j+1}$ if $j = 0$, and otherwise $d(\psi_{j+1})$ is obtained from $\psi_{j+1}$ by performing the following replacement exhaustively:
if a constraint $e_1\odot e_2$ in $\psi_{j+1}$ contains $v^{j+1}$ for some variable $v$, 
use a fresh variable $x$ of the same sort and replace $e_1\odot e_2$
by 
\[(e_1[v^{j+1} \mapsto x'] \odot e_2[v^{j+1} \mapsto x'] \wedge \wX (x^{j} = v)) \vee \wX^j\bot\]
Here $e[y \mapsto z]$ for an expression $e$ is the expression obtained from $e$ by
substituting all occurrences of $y$ by $z$.
After applying this transformation exhaustively, $\psi_j := d(\psi_{j+1})$ has lookahead $j$.

Let $x_{v,j}$ denote the fresh variable introduced to substitute $v^j$, for all $v\in V$ and $1<j\leq m$, and $X$ the set of all such variables.
For a trace $\tau$ of length $n$ over variables $V$, let $\tau'$ be obtained from $\tau$ by setting $\tau'(i)(v) = \tau(i)(v)$ for all $v\in V$ and $0\leq i < n$;
and $\tau'(i)(x_{v,j}) = \tau(i+j)(v)$ if $i+j < n$, and otherwise $\tau'(i)(x_{v,j}) := 0$.

(1) We show that $[\tau',i](e) = [\tau', i](e[v^o \mapsto x_{v,o}])$, for an arbitrary $i$ and expression $e$ with lookahead $j$ such that $i+j<n$, and lookahead $o>1$. 
This is straightforward to show by structural induction on $e$.
If $e$ is a constant or a variable different from $v^o$, both evaluations clearly coincide.
If $e = e_1 + e_2$ or $e = e_1 - e_2$, the statement follows from the induction hypothesis.
Otherwise, $e = v^o$ and $[\tau', i](e[v^o \mapsto x_{v,o}]) = [\tau', i](x_{v,o})$ and by definition of $\tau'$, $\tau'(i) = \tau(i+o)(v)$ holds.

(2) Next, we show that $\tau',i\models \psi_j$ iff $\tau',i\models d(\psi_j)$ for arbitrary $i$ and $\psi'$, by induction on $(n-i,\psi')$.
The interesting case if the one where $\psi_j$ is a constraint $c = e_1 \odot e_2$.
If $c$ has lookahead smaller than $j$, then $d(c) = c$, so there is nothing to show. Otherwise, suppose $c$ is replaced by $(d(c)\wedge \wX (x^{j} = v)) \vee \wX^j\bot$ (the case for multiple replaced variables is similar).
If $\tau',i\models c$ holds, either $i+j < n$ and $[\tau',i](e_1) \odot [\tau',i](e_2)$
 holds, or $i+j \geq n$. In the latter case, $\tau',i\models \wX^j\bot$, so $\tau',i\models d(\psi')$. Otherwise, 
 $d(c) = e_1[v^{j+1}\mapsto x_{v,j+1}'] \odot e_2[v^{j+1} \mapsto x_{v,j+1}']$, and the claim follows from (1).
Similarly, if $\tau',i\not \models c$, we must have $i+j < n$ and $[\tau',i](e_1) \odot [\tau',i](e_2)$ does not hold; then the claim follows again from (1).
If $\psi_j = \psi \wedge \psi'$ or $\psi_j = \neg \psi$, the claim follows from the induction hypothesis as the formula gets smaller.
If $\psi_j = \sX \psi$ then $\tau',i\models \psi'$ can only hold if $i<n-1$ and
$\tau,i+1 \models \psi$ from which the claim follows again by the induction hypothesis; and the other direction is similar.
A similar reasoning applies to $\psi_j = \psi \U \psi'$.

From (2) we conclude that $\tau',i\models \psi_m$ iff $\tau',i\models \psi_1$.
Since $\psi_m = \psi$ contains only the variables $V$ and $\tau$ and $\tau'$ coincide on $V$, we it follows that $\tau\models \psi$ iff $\tau'\models \psi_1$.

Note that the variable $x_{v,j}$ must be distinct for every $v$ and $j$, but obviously, if $v^j$ has multiple occurrences, the respective fresh variable and additional constraint can be shared.
\end{proof}

\subsection{Automata}

\begin{lemma}
\label{lem:lookback:lookahead}
For $\psi$ with lookahead 1,
$\tau \models \psi$ iff $\tau \models \psi_\back$.
\end{lemma}
\begin{proof}
	We extend the semantics to account for variables in $V_{prev}$ and $V_{curr}$: first, we impose $[\tau,i](v_{curr})=\tau(i)(v)$ $0{\leq} i{\leq} n$ and $[\tau,i](v_{prev}) = \tau(i{-}1)(v)$ for $0{<}i{\leq} n$. 
	%while for $i=0$ we determine that the expression is not well-defined. 
	We denote by $\models_b$ the satisfaction relation defined as $\models$ apart from the modification above. %, so that $\tau,i\models_b c$ whenever $c$ contains a non well-defined expression as before.   
	Given this, we show that for all constraints $c$, $\tau,i\models c$ iff $\tau,i\models_b back(c)$. %, where $back(c)$ denotes the constraint (atomic formula) obtained by transforming $c$ as in Sec.~\ref{sec:automata}. 
	If $c$  does not contain primed variables, the claim trivially holds since $back(c)=c$. 
	Otherwise, consider the case $c=v' \odot u$, so that $back(c) = \wX (v_{curr} \odot u_{prev})$. 
	($\Rightarrow$) 
	For $i{<}n$ we have that $\tau,i\models c$ implies 
	$[\tau,i](v')\odot [\tau,i](u)$ thus  $\tau(i+1)(v) \odot \tau(i)(u)$ which implies $[\tau,i{+}1](v)\odot [\tau,i](u)$ and hence $\tau,i{+}1\models_b v_{curr} \odot u_{prev}$, so that $\tau,i\models_b \wX ( v_{curr} \odot u_{prev})$. 
	For $i{=}n$, we simply note that $\tau,i\models c$ as this is not well-defined for $\tau$ and $i$, and $\tau,i\models_b back(c)$. 
	The same  reasoning can be applied for arbitrary constraints. 
	The opposite direction ($\Leftarrow$) is proven by the same steps in reversed order. 
	It thus follows that $\tau,i\models \psi$ iff $\tau,i\models_b \psi_{\mathit{back}}$, since $\psi_{\mathit{back}}$ differs from $\psi$ only at the level of atoms.   
\end{proof}

Let $\LL_\back$ denote the set of properties obtained from $\LL$ after the transformation to negation normal form, and after replacing lookahead by lookback.
Formally, $\LL_\back$ is defined by the following grammar, where $c\in \CC(V_\prev\cup V_\curr)$:\\
$\psi :: = 
c \mid \neg c \mid
\psi {\wedge} \psi \mid \psi {\vee} \psi \mid \wX \psi \mid
\sX \psi \mid \G \psi \mid \psi \U \psi
$

To refine the notion of consistency to the level of constraint sets with $\lambda$s, and ensure that $\lambda$s
occur in the right place, we define $\lambda$-consistency:
\begin{definition}
Let $\constr(\varsigma) = \varsigma \setminus \{\lambda, \neg \lambda\}$.
A symbol $\varsigma \in \Sigma_\lambda$ is \emph{$\lambda$-consistent with instant $i$} of trace $\langle\alpha_0, \dots, \alpha_{n-1}\rangle$ (or simply with $(\tau, i)$)
if 
(1) $0\,{\leq}\,i\,{<}\,n{-}1$ and $\lambda \not\in \varsigma$, or $i=n{-}1$ and $\neg \lambda \not\in \varsigma$, 
(2) if $i\,{=}\,0$ then $\varsigma$ does not contain $V_\prev$ and
$\alpha_0 \models constr(\varsigma)$; and if $i\,{>}\,0$ then $\combine{\alpha_{i-1}}{\alpha_{i}} \models \constr(\varsigma)$.
\end{definition}

Let a property $\psi \in \LL_\back$ have \emph{safe lookback} if
\begin{inparaitem}
\item
$\psi$ is a (negated) constraint that does not mention $V_\prev$,
\item $\psi = \wX\psi'$, $\psi = \sX\psi'$, or $\psi \in \{\top, \bot\}$,
\item $\psi = \neg \psi'$ or $\psi = \G\psi'$, and $\psi'$ has safe lookback, or
\item $\psi = \psi_1 \wedge \psi_2$, $\psi = \psi_1 \vee \psi_2$, or $\psi = \psi_1 \U \psi_2$, and
$\psi_1$ and $\psi_2$ have safe lookback.
\end{inparaitem}

We next show that $\delta$ is total in the sense that it provides a matching transition for every instant of a trace. 

\begin{lemma}%[$\delta$ is total]
\label{lem:delta:total}
Consider $\psi \in \LL_\back\cup \{\top,\bot\}$, a trace $\tau$ of length $n$, and $0\leq i < n$ such that if $i=0$ then $\psi$ has safe lookback.
Then there is some $(\inquotes{\psi'}, \varsigma) \inn \delta(\inquotes{\psi})$ 
such that 
$\varsigma$ is $\lambda$-consistent with $(\tau, i)$.
\end{lemma}
\begin{proof}
By structural induction on $\psi$. The claim is easy to check for every base
case of the definition of $\delta$, and in all other cases it follows from the induction hypothesis and the definition of $\owedge$ and $\ovee$.
\end{proof}

The next result show that $\delta$ preserves and reflects satisfaction of properties at trace instants.

\begin{lemma}
\label{lem:delta}
Let $\psi \in \LL_\back$,
$\tau = \langle\alpha_0, \dots, \alpha_{n-1}\rangle$ a trace, and $0\,{\leq}\,i\,{<}\,n$.
Then the following are equivalent:
\begin{compactitem}
\item
$\tau,i \models \psi$ and if $i=0$ then $\psi$ has safe lookback, and
\item
there is some $(\inquotes{\psi'}, \varsigma)\in \delta(\inquotes{\psi})$ such that\\
\noindent
\begin{tabular}{@{\ }r@{\ }p{8cm}}
$(a)$ & $\varsigma$ is $\lambda$-consistent with instant $i$ of $\tau$, and\\
$(b)$ & if $i\,{=}\,n{-}1$ then $\psi'\,{=}\,\top$, and otherwise $\tau,i{+}1 \models \psi'$.
\end{tabular}
\end{compactitem}
\end{lemma}
\begin{proof}
We prove both directions simultaneously by induction on $\psi$, using the definition of $\delta$.
\begin{compactitem}
\item
If $\psi = \top$ then it must be $\psi'=\top$ and $\varsigma=\emptyset$, so
both directions hold trivially.
If $\psi = \bot$ it must be $\psi'=\bot$, so neither $\tau,i\models \psi$
nor $\tau,i+1\models \psi'$ hold.
\item Let $\psi$ be a constraint $c\in \CC(V_\prev\cup V_\curr)$. 
$(\Longrightarrow)$
Suppose $\tau,i\models c$. For $(\inquotes{\top}, \{c\}) \in \delta(\inquotes{c})$, (b) holds because $\psi' = \top$, and (a) holds
because $\tau,i\models c$ and the assumption that $c$ has safe lookback if $i=0$ implies $\lambda$-consistency of $\{c\}$.

$(\Longleftarrow)$
Suppose there is some 
$(\inquotes{\psi'}, \varsigma)\in \delta(\inquotes{c})$ such that (a) and (b) hold. By definition of $\delta$, $\psi'$ is either $\top$ or $\bot$, but the latter can be excluded, so that $constr(\varsigma) = \{c\}$.
By $\lambda$-consistency, we have either $i=0$, $\varsigma$ does not mention $V_\prev$, and $\alpha_0 \models constr(\varsigma)$, or $i > 0$ and $\combine{\alpha_{i-1}}{\alpha_{i}} \models  constr(\varsigma) = c$, so $\tau,i\models c$, and $c$ is well-defined at instant $i$ of $\tau$, so safe lookback holds if $i=0$.
\item
Let $\psi = \neg c$ for a constraint $c\in \CC(V_\prev\cup V_\curr)$. 
$(\Longrightarrow)$ 
Suppose $\tau,i\models \neg c$, so $\tau,i \not\models c$. We have $(\inquotes{\top}, \{neg(c)\}) \in \delta(\inquotes{\neg c})$. Then (b) holds because $\psi' = \top$, and (a) holds
because $\tau,i\not\models c$ and the assumption that $c$ is well-defined implies 
$\tau,i \models neg(c)$, and implies $\lambda$-consistency.

$(\Longleftarrow)$ 
Let $(\inquotes{\psi'}, \varsigma)\in \delta(\inquotes{\neg c})$ such that (a) and (b) hold. By definition of $\delta$, $\psi'$ is either $\top$ or $\bot$, but the latter can be excluded, so that $constr(\varsigma) = \{neg(c)\}$.
By $\lambda$-consistency, either $i\,{=}\,0$, $\varsigma$ does not mention $V_\prev$, and $\alpha_0 \models neg(c)$, or $i\,{>}\, 0$ and $\combine{\alpha_{i-1}}{\alpha_{i}} \models neg(c)$, so $\tau,i\models neg(c)$, and safe lookback if $i=0$ follows from $\lambda$-consistency.
\item
We have $\tau, i \models \psi_1 \wedge \psi_2$ iff $\tau, i \models \psi_1$ and
$\tau, i \models \psi_2$.
By the induction hypothesis, this is the case iff there are some $(\inquotes{\psi_1'}, \varsigma_1) \in \delta(\inquotes{\psi_1})$ and $(\inquotes{\psi_2'}, \varsigma_2) \in \delta(\inquotes{\psi_2})$
such that $\varsigma_1$, $\varsigma_2$ are $\lambda$-consistent with $\tau$ at $i$;
and either $i<n$ and $\tau,i+1 \models \psi_1'$ and $\tau,i+1 \models \psi_2'$, or
$\psi_1' = \psi_2' = \top$.
We have $(\inquotes{\psi_1' \wedge \psi_2'}, \varsigma_1 \cup \varsigma_2) \in \delta(\inquotes{\psi_1 \wedge \psi_2})$,
where
$\varsigma_1 \cup \varsigma_2$ is $\lambda$-consistent with $\tau$ at $i$ iff both $\varsigma_1$ and $\varsigma_2$ are.
Moreover, if $i<n$ then $\tau,i+1 \models \psi_1'\wedge \psi_2'$ iff $\tau,i+1 \models \psi_1'$ and $\tau,i+1 \models \psi_2'$ hold; and if $i=n$ then $\psi_1' \wedge \psi_2' = \top$ iff $\psi_1' = \psi_2' = \top$.
In addition, $\psi_1 \wedge \psi_2$ has safe lookback iff this holds for both $\psi_1$ and $\psi_2$.
\item
We have $\tau, i \models \psi_1 \vee \psi_2$ iff $\tau, i \models \psi_1$ or
$\tau, i \models \psi_2$. Note that $\psi_1 \vee \psi_2$ has safe lookback iff this holds for both $\psi_1$ and $\psi_2$.
Wlog, assume the former.
By the induction hypothesis, this is the case iff there are some $(\inquotes{\psi_1'}, \varsigma_1) \in \delta(\inquotes{\psi_1})$
such that $\varsigma_1$ is $\lambda$-consistent with $\tau$ at $i$;
and either $i<n$ and $\tau,i+1 \models \psi_1'$, or
$\psi_1' = \top$.
By Lem.~\ref{lem:delta:total}, there is some $(\inquotes{\psi_2'}, \varsigma_2) \in \delta(\inquotes{\psi_2})$ such that $\varsigma_1$ is $\lambda$-consistent with $\tau$ at $i$;
We have $(\inquotes{\psi_1' \vee \psi_2'}, \varsigma_1 \cup \varsigma_2) \in \delta(\inquotes{\psi_1 \vee \psi_2})$,
where
$\varsigma_1 \cup \varsigma_2$ is $\lambda$-consistent with $\tau$ at $i$ iff both $\varsigma_1$ and $\varsigma_2$ are.
Moreover, if $i<n$ then $\tau,i+1 \models \psi_1'\vee \psi_2'$ iff $\tau,i+1 \models \psi_1'$ or $\tau,i+1 \models \psi_2'$ hold; and if $i=n$ then $\psi_1' \vee \psi_2' = \top$ iff $\psi_1' = \top $ or $\psi_2' = \top$.
\item
($\Longrightarrow$)
If $\tau, i \models \sX \psi$ then $i<n{-}1$ and $\tau, i+1 \models \psi$.
As $(\inquotes{\psi}, \{\neg \lambda\}) \in \delta(\inquotes{\sX \psi})$ and $\{\neg \lambda\}$ is
$\lambda$-consistent with $\tau$ at $i$ because $i<n-1$, the claim holds.
($\Longleftarrow$)
If $(\inquotes{\psi'}, \varsigma) \in \delta(\inquotes{\sX \psi})$ such that
$\varsigma$ is $\lambda$-consistent and $\tau, i+1 \models \psi'$ or $\psi' = \top$, by the 
definition of $\delta$ it must be $\psi' = \psi$ and $\varsigma = \{\neg\lambda\}$.
By $\lambda$-consistency, $i<n-1$, so $\tau, i+1 \models \psi$, and hence $\tau, i \models \sX \psi$.
Safe lookback holds by definition.
\item
($\Longrightarrow$)
If $\tau, i \models \wX \psi$ then $i=n-1$, or $i<n-1$ and $\tau, i+1 \models \psi$.
In the latter case, we reason as for $\sX$.
If $i=n$, we take $(\inquotes{\top}, \{\lambda\}) \in \delta(\inquotes{\wX \psi})$ and $\{\lambda\}$ is
$\lambda$-consistent with $\tau$ at $i$, so the claim holds.
($\Longleftarrow$)
If $(\inquotes{\psi'}, \varsigma) \in \delta(\inquotes{\wX \psi})$ such that
$\varsigma$ is $\lambda$-consistent and $\tau, i+1 \models \psi'$ or $\psi' = \top$, by the 
definition of $\delta$ it must be either $\psi' = \psi$ and $\varsigma = \{\neg\lambda\}$, or $\psi' = \top$ and $\varsigma = \{\lambda\}$.
In the former case, we reason as for $\sX$, otherwise, $\tau, i \models \wX \psi$ holds anyway.
\item
($\Longrightarrow$)
If $\tau, i \models \psi_1 \U \psi_2$, either (i) $\tau, i \models \psi_2$, or (ii) $i<n{-}1$,
$\tau, i \models \psi_1$, and $\tau, i+1 \models \psi_1 \U \psi_2$.
In case (i), by the induction hypothesis there is some  $(\psi', \varsigma') \in \delta(\inquotes{\psi_2})$
such that (a) and (b) hold.
By (a), $\varsigma'$ is $\lambda$-consistent with $\tau$ at $i$. By Lem.~\ref{lem:delta:total}, there are some  $(\psi_1', \varsigma_1) \in \delta(\inquotes{\psi_1})$ and $(\psi_2', \varsigma_2) \in \delta(\inquotes{\sX (\psi_1 \U \psi_2)})$ such that $\varsigma_1$ and $\varsigma_2$ are $\lambda$-consistent with $\tau$ at $i$. We have that $(\inquotes{\psi}, \varsigma) \in \delta(\inquotes{\psi_1 \U \psi_2})$ for $\psi = \psi' \vee (\psi_1' \wedge \psi_2')$ and $\varsigma = \varsigma' \cup \varsigma_1 \cup \varsigma_2$, and $\varsigma$ is $\lambda$-consistent with $\tau$ at $i$. Moreover, by (b) either $i<n{-}1$ and $\tau,i+1 \models \psi'$, or $i=n{-}1$ and 
$\psi' = \top$. Thus, also $i<n{-}1$ and $\tau,i+1 \models \psi$, or $i=n{-}1$ and 
$\psi = \top$ hold.
In case (ii) we have $i<n{-}1$. As safe lookback holds for $\psi_1$ if $i=0$, by the induction hypothesis there is some  $(\psi_1', \varsigma_1') \in \delta(\psi_1)$
such that (a) and (b) hold, so as $i<n{-}1$ it is $\tau, i+1 \models \psi_1'$ ($\star$). Moreover,  $(\inquotes{\psi_1 \U \psi_2}, \{\neg \lambda\}) \in \delta(\inquotes{\sX (\psi_1 \U \psi_2)})$ and $\{\neg \lambda\}$ is
$\lambda$-consistent with $\tau$ at $i$ because $i<n{-}1$.
By Lem.~\ref{lem:delta:total}, there is also some  $(\psi_2', \varsigma_2') \in \delta(\psi_2)$ such that $\varsigma_2'$ is $\lambda$-consistent with $\tau$ at $i$.
Now, we have $(\inquotes{\psi}, \varsigma) \in \delta(\inquotes{\psi_1 \U \psi_2})$ for $\psi = \psi_2' \vee (\psi_1' \wedge (\psi_1 \U \psi_2))$ and $\varsigma = \varsigma' \cup \varsigma_1 \cup \varsigma_2$.
Then $\varsigma$ as a union is also $\lambda$-consistent with $\tau$ at $i$, and $\tau, i+1 \models \psi$ because of ($\star$) and the assumption of case (ii).

($\Longleftarrow$) Suppose there is some $(\inquotes{\psi}, \varsigma) \in \delta(\inquotes{\psi_1 \U \psi_2})$ such that (a) and (b) holds.
By the definition of $\delta$, we can write $\psi$ and $\varsigma$ as
$\psi = \psi_2' \vee (\psi_1' \wedge \psi_3')$ and $\varsigma = \varsigma_1 \cup \varsigma_2 \cup \varsigma_3$, where $(\inquotes{\psi_2'}, \varsigma_2) \in \delta(\inquotes{\psi_2})$,
$(\inquotes{\psi_1'}, \varsigma_1) \in \delta(\inquotes{\psi_1})$, and
$(\inquotes{\psi_3'}, \varsigma_3) \in \delta(\inquotes{\sX (\psi_1 \U \psi_2)})$.
By (a), all of $\varsigma_1$, $\varsigma_2$, $\varsigma_3$ must be $\lambda$-consistent.
By (b), either $i<n{-}1$ and $\tau,i+1 \models \psi_2' \vee (\psi_1' \wedge \psi_3')$, or
$i=n{-}1$ and $\psi_2' \vee (\psi_1' \wedge \psi_3') = \top$.
We distinguish these two cases: if $i=n{-}1$ then either $\psi_2' = \top$ or $\psi_1' \wedge \psi_3' = \top$, but the latter is impossible
because it would imply $\psi_3' = \top$, but by $\lambda$-compatibility this is not possible if $i=n$.
So $\psi_2' = \top$. By the induction hypothesis, $\tau, i \models \psi_2$, and hence $\tau, i \models \psi_1 \U \psi_2$.
If $i<n$ then $\tau,i+1 \models \psi_2' \vee (\psi_1' \wedge \psi_3')$, so either
$\tau,i+1 \models \psi_2'$ or $\tau,i+1 \models \psi_1' \wedge \psi_3'$.
In the former case, by the induction hypothesis $\tau, i \models \psi_2$, and hence $\tau, i \models \psi_1 \U \psi_2$. In the latter case, $\tau,i+1 \models \psi_1'$ and $\tau,i+1 \models \psi_3'$.
By the induction hypothesis, $\tau,i \models \psi_1$. By the definition of $\delta$, we must have
$\psi_3' = \psi_1 \U \psi_2$. As $\tau,i \models \psi_1$ and $\tau,i+1 \models \psi_1 \U \psi_2$, we obtain again $\tau,i \models \psi_1 \U \psi_2$.
\item The case for $\G$ is similar.
\qedhere
\end{compactitem}
\end{proof}

Let a word be $\lambda$-consistent with a trace $\tau$ if it is $\lambda$-consistent with $\tau$ at all instants $i$.

\begin{lemma}
\label{lem:deltastar}
Let $\psi$ have safe lookback.
Then $\rho \models \psi$
iff there is a  word $w\in \Sigma_\lambda^*$ that is $\lambda$-consistent with a trace $\tau$ such that $\inquotes{\top} \in \delta^*(\inquotes{\psi},w)$.
\end{lemma}
\begin{proof}
($\Longrightarrow$)
Suppose that $\tau \models \psi$, i.e., $\tau,0 \models \psi$.
We show that, more generally, for all $i$, $0\,{\leq}\,i\,{<}\,n$,
and every property $\chi \in \LL_\back$,
if $\tau, i \models \chi$ and $\chi$ has safe lookback in the case where $i=0$, then 
there is a word $w_i = \langle\varsigma_i, \varsigma_{i+1}, \dots, \varsigma_{n-1}\rangle$
such that
$\inquotes{\top} \in \delta^*(\inquotes{\chi},w_i)$,
and 
$\varsigma_j$ is $\lambda$-consistent with $(\tau,j)$ for all $j$, $i\,{\leq}\,j\,{<}\,n$.
The proof is by induction on $n-i$.
In the base case where $i=n{-}1$, we assume that $\tau,n \models \chi$. 
By Lem. \ref{lem:delta} there is some $\varsigma_n$ such that
$(\inquotes{\top}, \varsigma_n)\in \delta(\inquotes{\chi})$,
and $\varsigma_n$ is $\lambda$-consistent with $(\tau,n)$.
For the induction step, assume $i\,{<}\,n-1$ and $\tau, i \models \chi$.
By Lem. \ref{lem:delta} there is some
$(\inquotes{\chi'}, \varsigma_i) \in \delta(\inquotes{\chi})$ such that
$\tau, i{+}1 \models \chi'$, and
$\varsigma_i$ is $\lambda$-consistent with $(\tau, i)$.
By the induction hypothesis, 
there is a word $w_{i+1} = \langle\varsigma_{i+1}, \dots, \varsigma_{n-1}\rangle$
such that
$\inquotes{\top} \in \delta^*(\inquotes{\chi'},w_{i+1})$,
and 
$\varsigma_j$ is $\lambda$-consistent with $(\tau,j)$ for all $j$, $i\,{<}\,j\,{<}\, n$.
Thus, we can define $w_{i} = \langle\varsigma_i,\varsigma_{i+1}, \dots, \varsigma_{n}\rangle$,
which satisfies
$\inquotes{\top} \in \delta^*(\inquotes{\chi},w_{i})$ and 
$\varsigma_j$ is $\lambda$-consistent with $\tau$ at instant $j$ for all $j$, $i\,{\leq}\,j\,{<}\, n$.
This concludes the induction step.
By assumption, $\tau,0 \models \psi$  holds, and
$\psi$ has safe lookback.
From the case $i = 0$ of the above statement, we obtain a word $w$
such that
$\inquotes{\top} \in \delta^*(\inquotes{\psi},w)$ and $w$ is
$\lambda$-consistent with $\tau$.

($\Longleftarrow$)
Let $w =\langle \varsigma_0, \dots, \varsigma_{n-1}\rangle$ be $\lambda$-consistent with $\tau = \langle\alpha_0,\dots, \alpha_{n-1}\rangle$.
If $\inquotes{\top} \in \delta^*(\inquotes{\psi},w)$, there must be properties
$\chi_0,\chi_1,\dots, \chi_{n}$ such that $\chi_0 = \psi$, $\chi_{n} = \top$,
and $(\inquotes{\chi_{i{+}1}}, \varsigma_{i}) \in \delta(\inquotes{\chi_i})$
for all $i$, $0\leq i < n$.
As $w$ is $\lambda$-consistent with $\tau$, 
$\varsigma_i$ is $\lambda$-consistent with 
$(\tau, i)$ for all $i$, $0\,{\leq}\,i\,{<}\,n$.
In order to show that $\tau \models \psi$, we verify that
$\tau, i \models \chi_i$
for all $i$, $0\,{\leq}\,i\,{<}\,n$.
The reasoning is by induction on $n-i$.
In the base case $i\,{=}\,n{-}1$.
Thus $\chi_{n} = \top$ and 
$(\inquotes{\chi_{n}}, \varsigma_{n}) \in \delta(\inquotes{\chi_{n-1}})$, so that
with $\lambda$-consistency it follows from
Lem. \ref{lem:delta} that $\tau, n-1 \models \chi_{n-1}$.
In the induction step $i\,{<}\,n{-}1$, and we assume by the induction hypothesis that $\tau, i{+}1 \models \chi_{i{+}1}$.
We have
$(\inquotes{\chi_{i{+}1}}, \varsigma_{i+1}) \in \delta(\inquotes{\chi_i})$, so
$\tau, i \models \chi_i$ follows again from Lem. \ref{lem:delta} using $\lambda$-consistency. This concludes the induction step, and the claim follows for the case $i\,{=}\,0$ because 
$\chi_0 = \psi$.
\end{proof}

\noindent
It remains to connect this key property of $\delta$ with the NFA $\NFA$ where $\lambda$ and $\neg \lambda$ are removed. To this end,
we next establish auxiliary facts about occurrences of $\lambda$'s,
which is straightforward to prove following \cite[Lem. A.5]{FMW22a}.

\begin{lemma}
\label{lem:delta:last}
Let $\psi \in \LL_\back$ and
$(\inquotes{\chi},\varsigma) \in \delta(\inquotes{\psi})$.
\begin{compactenum}
\item[(1)] If $\lambda \in \varsigma$
then $\chi = \top$ or $\chi = \bot$.
\item[(2)] Suppose $\neg \lambda \in \varsigma$, $\lambda \not\in\varsigma$,
and $\chi = \top$, and $\varsigma$ is $\lambda$-consistent with some step $i$ of some run $\rho$.
Then there is some $(\inquotes{\top},\varsigma') \in \delta(\inquotes{\psi})$
such that $\neg \lambda\not\in \varsigma'$ and $\varsigma'$ is $\lambda$-consistent with step $i$ of $\rho$ as well.
\item[(3)] If $\chi$ is not $\top$ or $\bot$ then
$\varsigma$ contains $\lambda$ or $\neg \lambda$.
\end{compactenum}
\end{lemma}

\lemmaNFA*
\begin{proof}
First, $\tau \models \psi$ iff 
$\tau \models \psi_\back$ according to Lem.~\ref{lem:lookback:lookahead}.
Next, note that $\psi_\back$ has safe lookback, as can be seen by inductive reasoning on $\psi$: If $\psi$ is a constraint $c$, then if $c$ has no lookahead, $\back(c)$ does not contain variables $v_\prev$, and otherwise $\back(c)=\wX c'$, has safe lookback by definition; and similar for $\neg c$.
For the remaining cases the claim follows from the induction hypothesis.

We now show that $\NFA[\psi_\back]$ accepts $w$ iff $\tau \models \psi_\back$.
($\Longrightarrow$)
Let $w = \varsigma_0 \varsigma_1 \cdots \varsigma_{n-1}$, and
$q_0 \goto{\varsigma_0} q_1 \goto{\varsigma_1} \dots \goto{\varsigma_{n-1}} q_{n}$  ($\star$) 
be the respective accepting run of $\NFA[\psi_\back]$.
By Def.~\ref{def:NFA}, there are $\varsigma_i'$, 
such that $\varsigma_i = \varsigma_i' \setminus \{\lambda, \neg \lambda\}$
for all $0\leq i < n$, and $w' = \varsigma_0' \varsigma_1' \cdots \varsigma_{n-1}'$
satisfies $\inquotes{\top} \in \delta^*(\inquotes{\psi},w')$.
By Lem.~\ref{lem:delta:last} (2) we can choose $\varsigma_{n-1}'$ such that
$\neg \lambda \not\in \varsigma_{n-1}'$, and $\varsigma_{n-1}'$ is consistent with $\tau$ at $n-1$.
Moreover, we can assume that none of $q_0, \dots, q_{n-1}$ is $\bot$ or $\top$, so
by Lem.~\ref{lem:delta:last} (1)and (3) we have $\lambda\not\in \varsigma_i'$ for $i<n-1$.
Since moreover $w$ is consistent with $\tau$, the augmented word $w'$ is $\lambda$-consistent with $\tau$.
Thus $\tau \models \psi_\back$ follows from Lem.~\ref{lem:deltastar}.

($\Longleftarrow$)
If $\tau \models \psi_\back$ then by Lem.~\ref{lem:deltastar} there is a 
word 
$w = \varsigma_0 \varsigma_1 \cdots \varsigma_{n}$ 
that is $\lambda$-consistent with $\tau$ such that $\inquotes{\top} \in \delta^*(\inquotes{\psi_\back},w)$, so $\NFA[\psi_\back]$ accepts $w$.
\end{proof}

\begin{lemma}
\label{lem:word:for:trace}
For every non-empty trace $\tau$ there is a unique word $w \in \Theta_\curr \Theta^*$ consistent with $\tau$.
\end{lemma}
\begin{proof}
Let $\tau = \langle\alpha_0, \dots, \alpha_{n-1}$, and $C$ be the constraints in $\psi_\back$.
Since for every constraint $c\in C_\curr$, either $\alpha\models c$ or $\alpha\models neg(c)$, and $\Theta_\curr$ contains all maximal satisfiable subsets of $C_\curr\cup C_\curr^-$, there must be exactly one subset $\varsigma_0\in \Theta_\curr$ such that
$\alpha \models \bigwedge \varsigma_0$.
Similarly, $\combine{\alpha_i}{\alpha_{i+1}}$ must model either $c$ or $neg(c)$
for all constraints in $c$, so there must be exactly one $\varsigma_{i}\in \Theta$
such that $\combine{\alpha_i}{\alpha_{i+1}} \models \bigwedge \varsigma_{i}$, for all $0<i \leq n-1$. Thus, the word $w = \langle \varsigma_0, \dots, \varsigma_{n-1}\rangle$
is uniquely defined and consistent with $\tau$.
\end{proof}

\lemmaDFANFA*
\begin{proof}
Let $\leq$ be the partial order on $\Sigma^n$ defined by the pointwise subset relation, i.e., $w \leq w'$ iff 
$w = \langle\varsigma_1,\dots,\varsigma_n\rangle$,
$w' = \langle\varsigma_1',\dots,\varsigma_n'\rangle$,
and $\varsigma_i \subseteq \varsigma_i'$ for all $1\leq i \leq n$ and $n\geq 0$.

($\Longrightarrow$)
We show the slightly more general statement $(\star)$  that if $\{q_0\} \gotos{w} P$ in $\DFA$ for some $w\in \Theta_\curr\Theta^* \cup \{\epsilon\}$ and $P\in\QQ$, 
then for all $q\in P$ there is some word $u \leq w$ such that $q_0 \gotos{u} q$ in $\NFA$.
The reasoning is by induction on $w$.
If $w$ is empty, we set also $u$ to the empty word and the statement is obvious.
In the inductive step, suppose $w = w'\varsigma$, so there is some
$P'$ such that $q_0 \gotos{w} P'$ and $\Delta(P',\varsigma) = P$.
By definition of $\Delta$, for all $q\inn P$, there is some $q'\inn P'$
such that $(q',\varsigma',q) \in \varrho$ and $\varsigma' \subseteq \varsigma$.
By the induction hypothesis, there is a word $u \leq w$ such that $q_0 \gotos{u} q'$ in $\NFA$. We have $u\varsigma' \leq w\varsigma$
and $q_0 \gotos{u\varsigma'} q$ in $\NFA$, so the claim holds.

Now suppose $\DFA$ accepts a word $w$ consistent with a trace $\tau$, so $\{q_0\} \gotos{w} P$ for a final state $P$.
As $P$ is final, it must contain a state $q_F$ which is final in $\NFA$. By $(\star)$, there is a word $u \leq w$ such that $q_0 \gotos{u} q_F$ in $\NFA$, so $u$ is accepted. As $w$ is consistent with $\tau$, $u\leq w$, and by definition of $\leq$, also $u$ is consistent with $\tau$.

($\Longleftarrow$)
We show the slightly more general statement $(\star\star)$ that 
if $q_0 \gotos{u} q$ in $\NFA$ for some $u\in \Sigma^*$ and $q \in Q$ then for all $w\in \Theta_\curr\Theta^* \cup \{\epsilon\}$ such that $u\leq w$ it holds that
$\{q_0\} \gotos{w} P$ in $\DFA$ for some $P$ such that $q\in P$.
The reasoning is by induction on $u$.
If $u$ is empty, we set also $w$ to the empty word and the statement is obvious.
In the inductive step, consider a word $u\varsigma$, so there is some $q'$ such that $q_0 \gotos{u} q'$ in $\NFA$ and $(q', \varsigma,q) \in \varrho$. By the induction hypothesis,
for all $w\in \Theta_\curr\Theta^* \cup \{\epsilon\}$ such that $u\leq w$ it holds that
$\{q_0\} \gotos{w} P'$ in $\DFA$ for some $P'$ such that $q'\in P'$.
Then, for all such $P'$ and $\varsigma'$ such that 
$\varsigma \subseteq \varsigma'$ we have $u\varsigma \leq w\varsigma'$. By definition of $\Delta$ it holds that
$P = \Delta(P', \varsigma')$ contains $q$, i.e., $\{q_0\}\gotos{w} P' \goto{\varsigma'} P$ in $\DFA$, so the claim holds.

Now suppose $\NFA$ accepts a word $u$ consistent with $\tau$, so
$q_0 \gotos{u} q_F$ for some $q_F \in Q_F$.
Then $u$ is well-formed and non-empty.
By Lem.~\ref{lem:word:for:trace} there is some $w\in \Theta_\curr\Theta^*$
consistent with $\tau$. By construction of $\Theta_\curr$ and $\Theta$ as set of
maximal satisfiable subsets, we must have $u\leq w$.
By $(\star\star)$, $\{q_0\} \gotos{w} P$ in $\DFA$ for some $P$ such that $q_F\in P$,
that is, $\DFA$ accepts $w$.
\end{proof}

\subsection{Finite summary}

\begin{example}
The constraint graphs $\CG_\psi(A)$ and $\CG_\psi(B)$ for the automaton $\DFA[\psi_\back]$ from
Ex.~\ref{exa:lookahead} is shown in Fig.~\ref{fig:cgs}.
\label{exa:cgs}
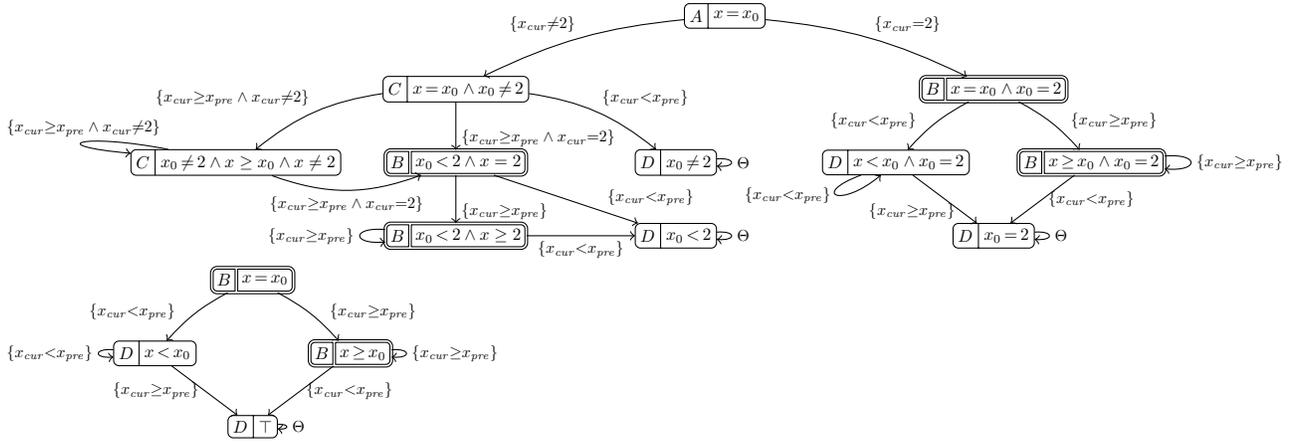
\begin{figure*}
\begin{tikzpicture}[node distance = 15mm]
\node[node] (0)  {\cgnode{$A$}{$x\,{=}\,x_0$}};
\node[node, below of=0, xshift=55mm, final] (1)
 {\cgnode{$B$}{$x\,{=}\,x_0 \wedge x_0\,{=}\,2$}};
\node[node, below of=0, xshift=-55mm] (2)
 {\cgnode{$C$}{$x\,{=}\,x_0 \wedge x_0\,{\neq}\,2$}};
\node[node, below of=1, xshift=20mm, final] (11)
 {\cgnode{$B$}{$x\,{\geq}\,x_0 \wedge x_0\,{=}\,2$}};
\node[node, below of=1, xshift=-20mm] (12)
 {\cgnode{$D$}{$x\,{<}\,x_0 \wedge x_0\,{=}\,2$}};
\node[node, below of=11, xshift=-20mm] (111)
 {\cgnode{$D$}{$x_0\,{=}\,2$}};
\node[node, below of=2, xshift=-45mm] (21)
 {\cgnode{$C$}{$x_0\,{\neq}\,2 \wedge x \geq x_0 \wedge x \neq 2$}};
\node[node, below of=2, final] (22)
 {\cgnode{$B$}{$x_0\,{<}\,2 \wedge x = 2$}};
\node[node, below of=2, xshift=45mm] (23)
 {\cgnode{$D$}{$x_0\,{\neq}\,2$}};
\node[node, below of=22, final] (221)
 {\cgnode{$B$}{$x_0\,{<}\,2 \wedge x \geq 2$}};
\node[node, below of=23] (222)
 {\cgnode{$D$}{$x_0\,{<}\,2$}};
\draw[goto, bend left=10] (0) to node[action, above, anchor=south west]{$\{\cur{x}{=}2\}$} (1);
\draw[goto, bend right=10] (0) to node[action, above, anchor=south east]{$\{\cur{x}{\neq}2\}$} (2);
\draw[goto, bend left=10] (1) to node[action, above, anchor=south west, near end]{$\{\cur{x}{\geq}\pre{x}\}$} (11);
\draw[goto, loop right, out=5, in=-5, looseness=7] (11) to  node[action, right]{$\{\cur{x}{\geq}\pre{x}\}$} (11);
\draw[goto, loop right, out=5, in=-5, looseness=7] (111) to  node[action, right]{$\Theta$} (111);
\draw[goto] (11) to node[action, right, anchor=west]{$\{\cur{x}{<}\pre{x}\}$} (111);
\draw[goto] (12) to node[action, left, anchor=east, near end]{$\{\cur{x}{\geq}\pre{x}\}$} (111);
\draw[goto, loop left, out=200, in=220, looseness=7] (12) to  node[action, left]{$\{\cur{x}{<}\pre{x}\}$} (12);
\draw[goto, bend right=10] (1) to node[action, above, anchor=south east, near end]{$\{\cur{x}{<}\pre{x}\}$} (12);
\draw[goto, bend right=15] (2) to node[action, above, anchor=south east]{$\{\cur{x}{\geq}\pre{x}\wedge \cur{x}{\neq}2\}$} (21);
\draw[goto] (2) to node[action, right, anchor= west, near end]{$\{\cur{x}{\geq}\pre{x}\wedge \cur{x}{=}2\}$} (22);
\draw[goto, bend left=15] (2) to node[action, above, anchor=south west]{$\{\cur{x}{<}\pre{x}\}$} (23);
\draw[goto, loop right, out=5, in=-5, looseness=7] (23) to  node[action, right]{$\Theta$} (23);
\draw[goto] (22) to node[action, above, anchor=south west, near end]{$\{\cur{x}{<}\pre{x}\}$} (222);
\draw[goto] (22) to  node[action, right, near end]{$\{\cur{x}{\geq}\pre{x}\}$} (221);
\draw[goto, loop right, out=175, in=185, looseness=7] (221) to  node[action, left]{$\{\cur{x}{\geq}\pre{x}\}$} (221);
\draw[goto] (221) to node[action, below]{$\{\cur{x}{<}\pre{x}\}$} (222);
\draw[goto, loop left, out=169, in=175, looseness=6] (21) to node[action, above]{$\{\cur{x}{\geq}\pre{x}\wedge \cur{x}{\neq}2\}$} (21);
\draw[goto, bend right=20] (21) to node[action, below]{$\{\cur{x}{\geq}\pre{x}\wedge \cur{x}{=}2\}$} (22);
\draw[goto, loop right, out=5, in=-5, looseness=7] (222) to  node[action, right]{$\Theta$} (222);
\end{tikzpicture}
\begin{tikzpicture}[node distance = 15mm]
\node[node, below of=0, xshift=52mm, final] (1)
 {\cgnode{$B$}{$x\,{=}\,x_0$}};
\node[node, below of=1, xshift=20mm, final] (11)
 {\cgnode{$B$}{$x\,{\geq}\,x_0$}};
\node[node, below of=1, xshift=-20mm] (12)
 {\cgnode{$D$}{$x\,{<}\,x_0$}};
\node[node, below of=11, xshift=-20mm] (111)
 {\cgnode{$D$}{$\top$}};
\draw[goto, bend left=10] (1) to node[action, above, anchor=south west, near end]{$\{\cur{x}{\geq}\pre{x}\}$} (11);
\draw[goto, loop right, out=5, in=-5, looseness=7] (11) to  node[action, right]{$\{\cur{x}{\geq}\pre{x}\}$} (11);
\draw[goto, loop right, out=5, in=-5, looseness=7] (111) to  node[action, right]{$\Theta$} (111);
\draw[goto] (11) to node[action, right, anchor=west]{$\{\cur{x}{<}\pre{x}\}$} (111);
\draw[goto] (12) to node[action, left, anchor=east]{$\{\cur{x}{\geq}\pre{x}\}$} (111);
\draw[goto, loop left, out=175, in=185, looseness=7] (12) to  node[action, left]{$\{\cur{x}{<}\pre{x}\}$} (12);
\draw[goto, bend right=10] (1) to node[action, above, anchor=south east, near end]{$\{\cur{x}{<}\pre{x}\}$} (12);
\end{tikzpicture}
\caption{\label{fig:cgs}Constraint graphs $\CG_\psi(A)$ and $\CG_\psi(B)$.}
\end{figure*}
\end{example}

\lemmacg*
\begin{proof}
(a) By induction on $w$. 
If $w$ is empty then $(P_0, \Cinit) = (P, \varphi)$, and since $h(w) = \Cinit$,
the claim holds.
For the inductive step, let $w = w'\cdot\langle\varsigma\rangle$ and consider a path
$(P_0, \Cinit) \gotos{w'} (P',\varphi') \goto{\varsigma} (P, \varphi)$.
By the induction hypothesis, 
$P_0 \gotos{w'} P'$ and $\varphi' \sim h(w')$.
By Def.~\ref{def:cg}, the edge $(P',\varphi') \goto{\varsigma} (P, \varphi)$ exists 
because $P' \goto{\varsigma}P$ in $\DFA$, and $\update(\varphi', \varsigma) \sim \varphi$. We thus have $P_0 \gotos{w} P$.
Moreover, by Def.~\ref{def:summary}, $\update(\varphi', \varsigma) \sim \update(h(w'), \varsigma)$, and as $\update(\varphi', \varsigma) \sim \varphi$ and $\update(h(w'), \varsigma) = h(w)$, we have $\varphi \sim h(w)$.

(b) By induction on $w$.
If $w$ is empty, $P_0 = P$. Then the empty path satisfies the claim.
For the inductive step, let $w = w'\cdot\langle\varsigma\rangle$ and suppose 
$P_0 \gotos{w'} P' \goto{\varsigma} P$ and $h(w)$ is satisfiable.
Then also $h(w')$ is satisfiable, and by the induction hypothesis
there is a path $(P_0, \Cinit) \gotos{w'} (P',\varphi')$ such that $\varphi' \sim h(w')$.
As $h(w) = \update(h(w'), \varsigma)$ is satisfiable, also 
$\update(\varphi', \varsigma)$ is satisfiable by Def.~\ref{def:summary},
so there must be an edge
$(P',\varphi') \goto{\varsigma} (P, \varphi)$ in the constraint graph such
that $\varphi \sim \update(\varphi', \varsigma)$.
Again by Def.~\ref{def:summary}, $\varphi' \sim h(w')$ implies
$\update(\varphi', \varsigma) \sim \update(h(w'), \varsigma)$, so $\varphi \sim h(w)$.
\end{proof}

\lemmaabstraction*
\begin{proof}
% We show the following slightly more general statement, which implies the claim:
% For a well-formed $w = \langle \varsigma_0, \dots, \varsigma_{n-1}\rangle \in \Sigma^*$,
% $\hist(w)$ is satisfied by assignment $\nu$ with domain $V \cup V_0$ iff
% there is a trace $\langle\alpha_0, \dots, \alpha_{n-1}\rangle$ 
% consistent with $w$ such that
% $\nu(v_0) = \alpha_0(v)$ and $\nu(v) = \alpha_n(v)$ for all $v\in V$.

($\Longrightarrow$)
Suppose $\hist(w)$ is satisfied by $\nu$.
We construct a respective trace $\tau$ by induction on $n$.
In the base case, $w$ is empty and $\hist(w) = \Cinit$, so for $\nu$ it must hold that $\nu(v_0) = \nu(v)$ for all $v\in V$.
For $\alpha$ the assignment with domain $V$ such that $\alpha(v)=\nu(v)$,
the trace $\langle\alpha\rangle$ satisfies the claim as it is consistent with $\langle\emptyset\rangle$.

In the induction step, $w = \langle \varsigma_0, \dots, \varsigma_{n}\rangle$ and we assume that
$\hist(w)$ is satisfied by $\nu$.
For $w' = \langle \varsigma_0, \dots, \varsigma_{n-1}\rangle$,
we have $\hist(w) = \update(\hist(w'), \varsigma_{n}) = \exists \vec U. \hist(w')(\vec U) \wedge \bigwedge \varsigma_{n}(\vec U, \vec V)$.
As $\nu \models h(w)$, 
there must be an assignment $\beta$ with domain $U\cup V \cup V_0$ that satisfies
$\hist(w')(\vec U) \wedge \bigwedge \varsigma_{n}(\vec U, \vec V)$ ($\star$)
such that $\beta$ coincides with $\nu$ on $V$ and $V_0$.
Let $\nu'$ be the assignment with domain $V\cup V_0$ such that 
$\nu'(\vec V) = \beta(\vec U)$ and $\nu'(\vec V_0) = \beta(\vec V_0)$,
which must satisfy $\hist(w')$.
By the induction hypothesis there is a trace $\tau' = \langle\alpha_0, \dots, \alpha_n\rangle$ consistent with $\langle\emptyset\rangle\cdot w'$ and $\nu'(v_0) = \alpha_0(v)$ and $\nu'(v) = \alpha_n(v)$ for all $v\in V$.
Let $\alpha_{n+1}$ be the assignment with domain $V$ such that $\alpha_{n+1}(v)=\nu(v)$ for all $v\in V$. By ($\star$), $\combine{\alpha_n}{\alpha_{n+1}} \models \bigwedge \varsigma_{n}$, so $\tau = \langle\alpha_0, \dots, \alpha_{n+1}\rangle$ is consistent with $\langle\emptyset\rangle\cdot w$ and
satisfies the claim because
$\nu(v_0) = \nu'(v_0) = \alpha_0(v)$ and $\nu(v) = \alpha_{n+1}(v)$ for all $v\in V$
by construction.

($\Longleftarrow$)
Suppose $\langle\alpha_0, \dots, \alpha_{n}\rangle$ is consistent with 
$\langle\emptyset\rangle\cdot w$.
We show that the assignment $\nu$ such that $\nu(v_0) = \alpha_0(v)$ and $\nu(v) = \alpha_n(v)$ for all $v\in V$ satisfies $h(w)$, by induction on the length $n$ of $w$.
If $n=0$, i.e., $w$ is empty, then $\hist(w) = \Cinit$.
As $\tau$ has length 1, $\alpha_0 = \alpha_n$, so the assignment
$\nu$ must satisfy $\nu(v_0) = \nu(v)$ for all $v\in V$, hence $\nu \models h(w)$ by definiiton of $\Cinit$.

In the induction step, let
$w = \langle \varsigma_0, \dots, \varsigma_{n}\rangle$ and suppose
$\tau = \langle\alpha_0, \dots, \alpha_{n+1}\rangle$
is consistent with $\langle\emptyset\rangle\cdot w$.
In particular, $\combine{\alpha_n}{\alpha_{n+1}} \models \bigwedge \varsigma_{n}$ ($\star$).
Thus, $\tau' = \langle\alpha_0, \dots, \alpha_{n}\rangle$ is 
consistent with $\langle\emptyset\rangle\cdot w'$ for $w' = \langle \varsigma_0, \dots, \varsigma_{n-1}\rangle$.
By the induction hypothesis, the assignment $\nu'$ with domain $V \cup V_0$
given by 
$\nu'(v_0) = \alpha_0(v)$ and $\nu'(v) = \alpha_n(v)$ for all $v\in V$
satisfies $h(w')$.
We have $h(w) = \update(\hist(w'), \varsigma_{n}) = \exists \vec U.\, \hist(w')(\vec U) \wedge \bigwedge \varsigma_{n}(\vec U, \vec V)$.
As $\nu' \models h(w')$, and because of ($\star$), assignment $\nu$ given by
$\nu(v_0) = \alpha_0(v)$ and $\nu(v) = \alpha_{n+1}(v)$ for all $v\in V$
satisfies $h(w)$.
\end{proof}

\subsection{Concrete criteria for solvability}

\theoremMC*
\begin{proof}
Our definition of history constraints (Def.~\ref{def:history constraint}) differs slightly from that in \cite{FMW22a} (due to a different notion of $\update$ and no fixed initial assignment).
However, the proof proceeds almost exactly as \cite[Thm.~5.2]{FMW22a} -- we sketch it here for completeness:
As all constraints in $\DFA$ are MC$_{\mathbb Q}$'s, there is a finite set $K$ of constants occurring therein. 
Let MC$_K$ be the set of quantifier-free boolean formulas where all atoms are MCs over variables $V_0 \cup V$ and constants in $K$, so MC$_K$ is finite up to equivalence.

We show that $\Phi := \mathcal Q \times \text{MC}_K$ is a finite history set, by showing that for any history constraint $h(w)$ of $\DFA$, there is a formula $\varphi \in \text{MC}_K$ s.t. $\varphi \equiv h(w)$, by induction on $w$.
If $w$ is empty, $h(w) = \Cinit \in \text{MC}_K$.
If $w = \langle\varsigma\rangle\cdot w'$ then $h(w) = \update(h(w'), \varsigma) = \exists \vec U.\,h(w')(\vec U) \wedge \chi$ for some $\chi \in \text{MC}_K$.
By induction hypothesis, $h(w')$ is equivalent to some $\varphi' \in \chi$.
Thus, $h(w) \equiv \exists \vec U.\,\varphi'(\vec U) \wedge \chi$, and
a quantifier elimination procedure \'{a} la Fourier-Motzkin ~\cite[Sec. 5.4]{KS16}
can produce a quantifier-free formula equivalent to the above that is again in MC$_K$.

As $\Phi$ is a finite history set, $(\Phi, \equiv)$ is a finite summary, so by Cor.~\ref{cor:main}, monitoring is solvable.
\end{proof}

\paragraph{Bounded lookback.}
Given a property $\psi$ with over $V_\curr \cup V_\prev$ and $\DFA[\psi]$ with states $P$, $P'$,
let $w=\langle\varsigma_0, \dots, \varsigma_{n-1}\rangle\in \Theta_\curr\Theta^*$ be a word such that $P \to^*_w P'$.
For each $0\,{\leq}\,i\,{<}\,n$, let $W_i$ be a fresh set of variables of the same size and sorts as $V$, and $\mathcal W = \bigcup_{i=0}^{n-1} W_i$.
The \emph{computation graph} $G_w$ of a word $w$ is the undirected graph with nodes $\mathcal W$ and an edge from $x\in W$ to $y\in W$ iff the two variables occur in a common literal of 
$\varsigma_0(\vec W_0)$, or  $\varsigma_i(\vec W_{i-1}, \vec W_i)$ for $0<i<n$.
The subgraph of $G_{w}$ of all edges corresponding to equality literals $x=y$ for $x, y \in \mc V$ is denoted $E_{w}$.
Moreover, we denote by $[G_{w}]$ the graph obtained from $G_{w}$ by collapsing
all edges in $E_{w}$.

\begin{definition}
$\DFA[\psi]$ has \emph{bounded lookback} if there is some $K$ such that for all
$w \in \Theta_\curr\Theta^*$
all acyclic paths in $[G_{w}]$ have length at most $K$.
\end{definition}

\begin{example}
\label{exa:computation graphs}
For $\psi = \wX (\G (\cur{x}{>}\pre{y} \wedge \cur{x}{=}\cur{y})$, a DFA $\DFA$ is shown below, where we abbreviate $c_1 = \cur{x}{>}\pre{y}$ and $c_2 = \cur{x}{=}\cur{y}$.\\
\resizebox{\columnwidth}{!}
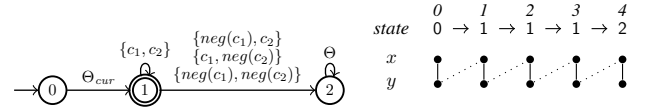
{
\begin{tikzpicture}[node distance=20mm]
 \node[state, minimum width=6mm] (0) {$0$};
 \node[state, minimum width=6mm, right of =0, final] (1) {$1$};
 \node[state, right of=1, minimum width=6mm, xshift=20mm] (2) {$2$};
\draw[edge] ($(0) + (-.5,0)$) -- (0);
 \draw[edge] (0) -- node[action, above] {$\Theta_\curr$} (1);
 \draw[edge] (1) -- node[action, above] {\begin{tabular}{c}$\{neg(c_1), c_2\}$\\$\{c_1, neg(c_2)\}$\\$\{neg(c_1), neg(c_2)\}$\end{tabular}} (2);
\draw[edge] (1) to[loop above, looseness=6] node[action, above] {$\{c_1, c_2\}$} (1);
\draw[edge] (2) to[loop above, looseness=6] node[action, above] {$\Theta$} (2);
\begin{scope}[xshift=50mm, xscale=.6, yscale=.8]
\node[scale=.7] at (-1,.5) {$x$};
\node[scale=.7] at (-1,.1) {$y$};
\node[scale=.7] at (-1,1) {state};
% \node[scale=.75] at (-1,1.3) {instant};
\foreach \i/\l in {0/0,1/1,2/1,3/1,4/2} {
  \node[scale=.65] at (\i,1.3) {\i};
  \node[scale=.65] (state\i) at (\i,1) {$\m\l$};
  \node[fill, circle, inner sep=0pt, minimum width=1mm] (x\i) at (\i,.5) {};
  \node[fill, circle, inner sep=0pt, minimum width=1mm] (y\i) at (\i,.1) {};
  }
\foreach \s/\t in {0/1, 1/2, 2/3, 3/4} {
  \draw[->, shorten >=.6mm, shorten <=.6mm] (state\s) -- (state\t);
  }
\draw (x0) -- (y0);
\draw (x1) -- (y1);
\draw (x2) -- (y2);
\draw (x3) -- (y3);
\draw (x4) -- (y4);
\draw[dotted] (y0) -- (x1);
\draw[dotted] (y1) -- (x2);
\draw[dotted] (y2) -- (x3);
\draw[dotted] (y3) -- (x4);
\end{scope}
\end{tikzpicture}
}
The computation graph for the word
$w = \{c_2\},\{c_1, c_2\}, \{c_1, c_2\}, \{c_1, neg(c_2)\}$ is shown on the right,
where equality edges are drawn solid, and other edges dotted.
In $G_w$ the maximal path length is 9, while in $[G_w]$ it is 4 (i.e., after collapsing equality edges).
However, as the loop in state 1 can be executed an unbounded number of times, the path length is unbounded as the pattern in the above $G_w$ repeats.
Thus $\DFA$ does not have bounded lookback, and neither does $\psi$.
\end{example}

\theoremboundedlookback*
\begin{proof}
Let $\psi$ have $K$-bounded lookback, for some $K>0$, and $\Phi$ be the set of formulas with free variables $V\cup V_0$, 
quantifier depth at most $K$, and using as atoms all constraints in 
$\DFA$, but where variables in $V_\curr\cup V_\prev$ may be replaced by $V$ or any of the quantified variables.
Being a set of formulae with bounded quantifier depth over a finite vocabulary, $\Psi$ is finite up to equivalence.
To prove that $\mathcal Q \times\Phi$ is a history set, 
we show that for every $w\in \Theta_\curr\Theta^*$, $\Phi$ contains a formula $\varphi$ such that $\varphi \equiv h(w)$.
It then follows that $(\Phi, \equiv)$ is a finite summary.
The proof is by induction of $w$.
If $w$ is empty,  $h(w) = \Cinit$ has quantifier depth 0, so there is 
a formula in $\Phi$ equivalent to $h(w)$ by definitio of $\Phi$.
Otherwise, $w = w'\cdot\langle\varsigma\rangle$, and
$h(w) = \update(h(w'), \varsigma)$.
By induction hypothesis there is some $\varphi' \in \Phi$ with $\varphi' \equiv h(w')$.
By definition of $\update$, we have $h(w) \equiv \exists \vec U. \varphi'(\vec U) \wedge \chi$ for some quantifier free formula  $\chi\in \Phi$.
Let $[\varphi]$ be obtained from $\varphi$ by eliminating all equality literals $x=y$ in $\varphi$ and uniformly substituting all variables in an equivalence class by some arbitrary representative.
Since $\DFA$ has $K$-bounded lookback, and $[\varphi]$ encodes a part of
$[G_{w}]$, 
$[\varphi]$ is equivalent to a formula $\phi$ of quantifier depth at most $K$ over the same vocabulary, 
obtained by dropping irrelevant literals and existential quantifiers, so that $\phi \in \Phi$.

Since $\psi$ admits a finite summary, monitoring is solvable by Cor.~\ref{cor:main}.
\end{proof}\

\begin{example}
\label{exa:bounded}
Consider a reactive system with two sensors $x$ and $y$ that are updated in each cycle.
The system has two clients who issue requests if conditions $x+y \leq 5$ and $2x - y > 7$ are satisfied, respectively.
Additionally, values of $y$ are non-negative, and it is known that if at some point $x \geq 5$ then $x$ will never be lower than this threshold in the future. In order to check whether the two requests can be issued at the same time, the conjunction of the 
following properties can be monitored:
$\G(x+y \leq 5 \to \wX \mathit{req}_1)$,
$\G(2x-y > 7 \to \wX \mathit{req}_2)$,
$\G (y \geq 0)$,
$\G (x\geq 5 \to x' \geq 5)$, and
$\F (\mathit{req}_1 \wedge \mathit{req}_2)$
Note that the variables $\mathit{req}_1$ and $\mathit{req}_2$ are here used as booleans, but they can easily be modeled as integers.
This property has 1-bounded lookback, because in all atoms with multiple variables, comparisons are only among variables at the same instant. Thus, monitoring is solvable and can be used to e.g. detect that once $x$ is larger than 5, the property is permanently violated.
\end{example}

\theoremMCZ*
\begin{proof}
We modify our monitoring approach as indicated in the proof sketch, and describe here only the missing details:
Let $\alpha := \alpha_n$ be the last assignment of a given trace, as in procedure \textsc{Monitor}.
Let $\mathcal K$ be the finite set of integers that contains all constants occuring in $\psi$ and $\alpha$, as well as $0$.
We use the finite summary $(\text{GC}_K, \sim_K)$ described in \cite{FMW22a} when computing constraint graphs, where $K = max \{k-k' \mid k, k'\in \mathcal K\}$.
When computing CGs in line 5 of the procedure, we start from the formula $\varphi_\alpha := \bigwedge v = \alpha(v)$ instead of $\Cinit$. 
% By \cite[Thm.~5.5]{FMW22a} a finite summary exists, so the CG computation will terminate.
For the correctness proof, the notion of a history constraints changes such that
for the empty word $w$, $h(w) = \varphi_\alpha$.
Thus, history constraints are formulas with free variables $V$ (the variables $V_0$ are no longer necessary).
Lem.~\ref{lem:abstraction} changes to the following statement:
For $w\in \Theta^*$ of length $k$,
$\hist(w)$ is satisfied by assignment $\nu$ with domain $V$ iff
$\langle\emptyset\rangle\,{\cdot}\,w$ is consistent with a trace $\alpha_0, \dots, \alpha_{k}$ 
such that $\nu(v_0)\,{=}\, \alpha(v)$ and $\nu(v)\,{=}\, \alpha_k(v)$ for all $v\in V$.
That is, the statement is as before, but the initial assignment is fixed to $\alpha$.
The notions of a history set and a summary as well as Lem.~\ref{lem:cg}, are as before, but with the updated notion of history constraints.
By reasoning as in \cite[Thm.~5.5]{FMW22a} a finite summary exists, so the CG computation will terminate.
The statement of Thm.~\ref{thm:lookahead1:monitoring} changes as follows:\\
\emph{For a DFA $\DFA$ for $\psi$ and a trace $\tau$, let $w\in \Theta_\curr\Theta^*$ be
the word consistent with $\tau$ and
$P$ the $\DFA$ state
s.t. $\{q_0\} \gotos{w} P$. For $G:= \CG_\psi(P, \sim_K)$ the constraint graph from $P$,
\begin{compactitem}
\item
if $P\,{\in}\,P_F$ then $\tau \models \llbracket \psi{=}\CS \rrbracket$ if
$G$ has a non-empty path to a non-final state, and $\tau \models \llbracket \psi{=}\PS \rrbracket$ otherwise,
\item
if $P\,{\not\in}\,P_F$ then $\tau \models \llbracket \psi{=}\CV \rrbracket$ if
$G$ has a non-empty path to a final state, and $\tau \models \llbracket \psi{=}\PV \rrbracket$ otherwise.
\end{compactitem}}
The proof of the theorem is almost as before. E.g. for the first case, if $P\,{\in}\,P_F$ and $G$ has a path to a non-final state $(P',\varphi)$, then by Lem.~\ref{lem:cg} $P \to_u^+ P'$ and $\varphi \sim h(u)$  is satisfiable. By the modified Lem.~\ref{lem:abstraction} there is a trace consistent with $u$ starting with $\alpha$. We can build the extension $\tau''$ of $\tau$ as in the proof of Thm.~\ref{thm:lookahead1:monitoring}, which is consistent with $wu$ and leads to the non-accepting state $P'$, so by the results on automata $\tau'' \not\models \psi$ and hence $\tau \models \llbracket \psi{=}\CS \rrbracket$.
\end{proof}
}{}

\end{document}